\documentclass[a4paper,USenglish,cleveref,autoref]{lipics-v2021}

\pdfoutput=1 %
\hideLIPIcs

\usepackage{graphicx}

\usepackage{xcolor}

\usepackage{amsmath}
\usepackage{amssymb}
\usepackage{mathtools}
\usepackage{nicefrac}
\usepackage{microtype}

\makeatletter
\newenvironment{ilp}[1]
 {%
  \addtocounter{equation}{-1}%
  \crefalias{subequations}{ilp}
  \begin{subequations}%
  \def\@currentlabel{#1}%
  \renewcommand{\theequation}{#1.\alph{equation}}%
 }
 {\end{subequations}}
\makeatother %

\usepackage{csquotes} %
\usepackage{marvosym} %

\colorlet{mygreen}{green!70!black}
\colorlet{myblue}{cyan!90!black}
\colorlet{myorange}{orange!75!white}

\usepackage{xurl} %
\Crefname{lemma}{Lemma}{Lemmas}
\Crefname{observation}{Observation}{Observations}
\Crefname{corollary}{Corollary}{Corollaries}
\Crefname{definition}{Definition}{Definitions}
\Crefname{case}{Case}{Cases}
\Crefname{ilp}{ILP}{ILPs}
\crefname{ilp}{ILP}{ILPs}

\usepackage[abbreviations,shortcuts]{glossaries-extra}

\newglossary*{problem}{Problems}
\newcommand{\newproblem}[3]{
    \newglossaryentry{pr:#1}{type=problem,
        name={\ensuremath{#2}},
        description={#3},
        sort={#1}
    }
}
\newcommand{\prob}[1]{\glsentryname{pr:#1}}
\newcommand{\probf}[1]{\glsentrydesc{pr:#1} (\glsentryname{pr:#1})}
\newcommand{\probl}[1]{\glsentrydesc{pr:#1}}

\newabbreviation{DAG}{DAG}{directed acyclic graph}
\newabbreviation{DSP}{DSP}{directed series-parallel graph}
\newabbreviation{EAS}{EAS}{edge-anchored subgraph}
\newabbreviation{MEAS}{MEAS}{maximal edge-anchored subgraph}
\newabbreviation{FPTAS}{FPTAS}{Fully Polynomial Time Approximation Scheme}
\newabbreviation{FPT}{FPT}{Fixed Parameter Tractable}
\newabbreviation{MED}{MED}{minimum equivalent digraph}
\newabbreviation{LSP}{LSP}{laminar series-parallel graph}
\newabbreviation{MCF}{MCF}{multi-commodity flow}
\newabbreviation{MCPS}{MCPS}{minimum capacity-preserving subgraph}
\newabbreviation{ILP}{ILP}{integer linear program}

\newproblem{DHC}{\textsf{DHC}}{\textsf{Directed Hamiltonian Cycle}}
\newproblem{UHC}{\textsf{UHC}}{\textsf{Undirected Hamiltonian Cycle}}
\newproblem{MCPS}{\textsf{MCPS}}{\textsf{Minimum Capacity-Preserving Subgraph}}
\newproblem{MCPS*}{\textsf{MCPS}^*}{\textsf{Minimum Capacity-Preserving Subgraph, Altered}}
\newproblem{MGCPS}{\textsf{MGCPS}}{\textsf{Minimum Global Capacity-Preserving Subgraph}}
\newproblem{MUGCPS}{\textsf{UGCPS}}{\textsf{Minimum Undirected Global Capacity-Preserving Subgraph}}
\newproblem{MSCS}{\textsf{MSCS}}{\textsf{Minimum Strongly Connected Subgraph}}
\newproblem{MED}{\textsf{MED}}{\textsf{Minimum Equivalent Digraph}}
\newproblem{SC}{\textsf{SC}}{\textsf{Set Cover}}
\newproblem{VC}{\textsf{VC}}{\textsf{Vertex Cover}}
\newproblem{MST}{\textsf{MST}}{\textsf{Minimum Spanning Tree}}
\newproblem{DSND}{\textsf{DSND}}{\textsf{Directed Survivable Network Design}}
\newproblem{MMCFS}{\textsf{MMCFS}}{\textsf{Minimum Multi-Commodity Flow Subgraph}}

\let\union\cup
\let\intersect\cap
\let\Union\bigcup

\newcommand{\ecap}{\ensuremath{\textit{cap}}}
\newcommand{\cost}{\ensuremath{\textit{cost}}}
\newcommand{\comm}{\ensuremath{\tilde{e}}} %
\newcommand{\stretcha}{\beta}
\newcommand{\mcfalpha}{\alpha}
\newcommand{\ssep}{\mid}

\newcommand{\medval}{\ensuremath{\textit{med}}}

\newcommand{\univ}{\ensuremath{U}}
\newcommand{\setcover}{\ensuremath{\mathcal{C}}}
\newcommand{\sets}{\ensuremath{\mathcal{S}}}

\newcommand{\bigO}{\ensuremath{\mathcal{O}}}
\newcommand{\Q}{\ensuremath{\mathbb{Q}}}

\newcommand{\N}{\ensuremath{\mathbb{N}}}
\newcommand{\be}{\ensuremath{\coloneqq}}
\newcommand{\transpose}[1]{\ensuremath{#1^{\top}}}
\newcommand{\meancap}[1]{\ensuremath{\overline{\ecap_{#1}}}}
\newcommand{\flow}{\ensuremath{f}}
\newcommand{\flowsum}{\ensuremath{\mathbf{f}}}
\newcommand{\flowset}{\ensuremath{\mathcal{F}}}
\newcommand{\demand}{\ensuremath{T}}
\newcommand{\alldemand}{\ensuremath{\mathbb{T}_E}}
\newcommand{\lpfeas}{\ensuremath{x}}
\newcommand{\lpopt}{\ensuremath{x^*}}

\newcommand{\algval}{\ensuremath{z}}
\newcommand{\lpval}{\ensuremath{z^*}}
\newcommand{\ilpval}{\ensuremath{z^*_\mathit{int}}}

\DeclarePairedDelimiter\ceil{\lceil}{\rceil}

\usepackage{tikz}
\usetikzlibrary{arrows,backgrounds,shapes,calc,decorations,decorations.pathreplacing,decorations.pathmorphing,positioning}
\tikzset{main_edge/.style={line width=2.3pt}}
\tikzset{overlap/.style={draw=magenta}}
\tikzset{dot/.style={circle,fill=black,inner sep=1pt,minimum size=1pt}}
\tikzset{edge/.style={
    rounded corners, -stealth, draw=black, very thick, shorten <= 2pt,shorten >= 2pt}}

\title{Traffic-Oblivious Multi-Commodity Flow Network Design}

\author{Markus Chimani}{Theoretical Computer Science, Osnabrück University,
Osnabrück, Germany}{markus.chimani@uos.de}{https://orcid.org/0000-0002-4681-5550}{}%

\author{Max Ilsen}{Theoretical Computer Science, Osnabrück University,
Osnabrück, Germany}{max.ilsen@uos.de}{https://orcid.org/0000-0002-4532-3829}{}

\authorrunning{M. Chimani and M. Ilsen} %

\Copyright{Markus Chimani and Max Ilsen} %

\ccsdesc[500]{Mathematics of computing~Network flows}
\ccsdesc[500]{Mathematics of computing~Approximation algorithms}
\ccsdesc[300]{Networks~Network design and planning algorithms}

\keywords{Multi-commodity flow, Digraphs, LP-rounding, Approximation algorithm} %

\category{} %

\relatedversion{} %

\funding{Funded by the Deutsche Forschungsgemeinschaft (DFG, German Research Foundation) grant 461207633 (CH 897/7-2).}

\nolinenumbers %

\begin{document}
\maketitle

\begin{abstract}
    We consider the \probf{MMCFS} problem:
    given a directed graph~$G$ with edge capacities~$\ecap$ and a
    retention ratio~$\mcfalpha\in(0,1)$, find an edge-wise minimum subgraph~$G' \subseteq
    G$ such that for all traffic matrices~$T$ routable in $G$ using a \acl{MCF},
    $\mcfalpha\cdot T$ is routable in~$G'$.
    This natural yet novel problem is motivated by recent research that
    investigates how the power consumption in backbone computer networks can be
    reduced by turning off connections during times of low demand without
    compromising the quality of service.
    Since the actual traffic demands are generally not known beforehand,
    our approach must be traffic-oblivious, i.e., work for all possible sets
    of simultaneously routable traffic demands in the original network.

    In this paper we present the problem, relate it to other known problems in
    literature, and show several structural results, including a reformulation,
    maximum possible deviations from the optimum, and NP-hardness (as well as a
    certain inapproximability) already on very restricted instances.
    The most significant contribution is a
    $\max(\nicefrac{1}{\mcfalpha}, 2)$-approximation
    based on a surprisingly simple LP-rounding scheme.
    We also give instances where this worst-case approximation ratio is met and
    thus prove that our analysis is tight.
\end{abstract}

\newpage
\section{Introduction}

We present the (suprisingly seemingly novel) network design problem \probf{MMCFS}:
given a directed flow network~$G$ with edge capacities and a \emph{retention
ratio}~$\mcfalpha\in(0,1)$, find a subnetwork~$G'\subseteq G$ of minimum size
such that $G'$ still allows for a \ac{MCF} routing of \emph{any} traffic demands
routable in $G$ when they are scaled down by factor~$\mcfalpha$.

The problem arises naturally in recent research concerning power saving in backbone
(Tier~1) networks of Internet service providers. There,
the overall amount of traffic
has distinct peaks in the evenings (when people are, e.g., streaming videos) and
lows late at night and in the early
mornings~\cite{DBLP:journals/ton/SchullerACHS18}. Clearly, the networks are built to handle the peak times.
This opens up the possibility to reduce the power consumption of the network by
turning off some resources---e.g., connections, line cards, or servers---during low traffic
periods~\cite{DBLP:conf/icnp/ZhangYLZ10,DBLP:conf/icc/ChiaraviglioMN09}.

So consider a computer network that allows for the simultaneous routing of the
traffic at peak times.
This traffic is comprised of a set of \emph{commodities}, where each commodity
is identified by a pair of nodes~$(s,t)$ in the network and has a
\emph{demand}~$d$ specifying that $d$~flow units have to be sent from~$s$
to~$t$.
The entirety of the demands for each pair of nodes is encoded in a
\emph{traffic matrix}~$\demand$. %
Commonly, one makes the simplifying assumption that during
low traffic periods the traffic demands are upper bounded by a down-scaling of
$\demand$ using a factor~$\mcfalpha\in(0,1)$; the most practically relevant
scenarios concern $\mcfalpha \geq \frac{1}{2}$~\cite{DBLP:conf/lcn/OttenICA23}.
The task now is to minimize the size of the network by deactivating connections
such that the reduced network still accommodates a routing of the scaled-down
demands~$\mcfalpha\cdot\demand$.
However, in practice the traffic matrix~$\demand$ may change from day to day (in
fact, even within sampling windows of 15 minutes) and while we can assume the
capacity of the network to be large enough for all occurring traffic at any
given time, $\demand$~is usually not known beforehand. Thus, our solution should
be \emph{traffic-oblivious}, i.e., independent of any specific traffic matrix.

Technically, routing in realistic scenarios is \emph{not} done via fully general
\aclp{MCF}: while \ac{MCF} would be optimal in terms of minimizing
congestion~\cite{DBLP:journals/ton/SchullerACHS18}, it would be too complicated
and temporally unstable for the routing hardware.
Instead, simpler techniques like \emph{2-segment routing} are used (which, in
contrast to trivial shortest path routing, routes flow along a sequence of two
chained shortest subpaths)~\cite{DBLP:journals/rfc/rfc8402}.
Interestingly, studies show that in realistic networks, these routing solutions
are virtually identical to those achieved by
\ac{MCF}~\cite{DBLP:journals/ton/SchullerACHS18,DBLP:conf/infocom/BhatiaHKL15}.
At the same time, for a given fixed network and traffic matrix, an \ac{MCF}
can be computed in polynomial time while establishing an optimal routing table
for 2-segment routing (which is then deployable on router hardware) is
NP-hard~\cite{DBLP:conf/cp/HartertSVB15}.
Thus, we describe the feasibility of solutions in our problem
setting in terms of routability via \ac{MCF}, and assume that realistic
(simpler) routing protocols will still be able to attain effective routability.

Let us give a formal description of our problem:
Given a directed graph (or \emph{digraph})~$G=(V,E)$ with positive edge
capacities~$\ecap \colon E~\to~\Q$ and a traffic matrix~$\demand$, a \emph{flow}~$\flow_{s,t}
\colon E~\to~\Q$ from a vertex~$s\in V$ to a vertex~$t \in V$ is a function
satisfying the flow conservation constraints
\begin{align*}
    \sum_{uv \in E} \flow_{s,t}(uv) - \sum_{vu \in E} \flow_{s,t}(vu) &\; = \;
            \begin{cases*}
                -\demand(s,t) & if $v=s$\\
                \demand(s,t) & if $v=t$\\
                0 & else
            \end{cases*}
            & &\forall v\in V.
        \intertext{
    A \emph{\acf{MCF}} is a set of flows~$\flowset = \{\flow_{s,t} \ssep (s,t) \in V^2\}$
    satisfying}
    \sum_{(s,t) \in V^2} \flow_{s,t}(uv) & \leq \ecap(uv) & & \forall uv \in E.
\end{align*}
For an edge $e=st\in E$, we may use the shorthand notation $\demand(e) \be \demand(s,t)$.
Further, %
we call a traffic matrix~$\demand$ \emph{routable in an edge set} $A \subseteq
E$ if there exists an \ac{MCF}~$\flowset$ for~$(G,\ecap,\demand)$ with
$\sum_{\flow \in \flowset} \flow(e) = 0$ for all~$e \notin A$. %
Based on this notion, we define the \probf{MMCFS} problem as follows:
given a digraph~$G=(V,E)$ with edge capacities~$\ecap$ and a retention
ratio~$\mcfalpha\in(0,1)$, find the edge set~$A \subseteq E$ with minimum
cardinality such that for all traffic matrices~$T$ that are routable in~$E$,
$\mcfalpha\cdot\demand$ is routable in~$A$.

Throughout this paper, when mentioning a problem's name or the corresponding
abbreviation in sans-serif typeface, e.g.\ \prob{MMCFS}, we refer to the
optimization question.
To indicate that a subgraph is an optimal solution for the problem, we will
denote it by the abbreviation in normal typeface, e.g.\ an MMCFS~$(V,A)$ with $A\subseteq E$.

\subparagraph{Our contribution.}
We present the \prob{MMCFS} problem, which has both a natural formulation and
practical applicability. After discussing related problems from literature
in \Cref{sec:relwork}, we give some structural results in \Cref{sec:structural}:
We show how the \prob{MMCFS} problem, even though it is
traffic-oblivious in nature, can be reformulated to consider a
specific single \enquote{hardest} traffic matrix. We also establish how an \ac{MCF}
can be routed in an optimal \prob{MMCFS} solution, and how the ratio between the
values of a feasible \prob{MMCFS} solution and an optimal one relates to their
average edge capacities.
In \Cref{sec:complexity}, we prove that \prob{MMCFS} is
NP-hard already with unit edge capacities.
Additionally, we show that it is NP-hard (and a closely related problem cannot
be approximated within a sublogarithmic factor) already on \acp{DAG}.
Our most important contribution is given in \Cref{sec:mcfs_approx}, where we
present a $\max(\nicefrac{1}{\mcfalpha}, 2)$-approximation algorithm for
\prob{MMCFS}:
after modelling \prob{MMCFS} as an ILP, we can deduce a surprisingly simple
LP-rounding scheme, whose complexity is solely shifted to the correctness proof.
Moreover, we show that our analysis of this algorithm is tight.

\section{Related Work}\label{sec:relwork}
There is a rich body of work on \aclp{MCF}---see e.g.\ \cite[Ch.\
17]{DBLP:books/daglib/0069809} for a primer on this topic and
\cite{DBLP:journals/or/SalimifardB22} for a recent literature review.
The ability to route an \ac{MCF} in an \prob{MMCFS} solution not only determines
the latter's feasibility, the problem of routing an \ac{MCF} also has many
close ties to several other network design problems.
These, however, involve constraints unrelated to \prob{MMCFS}, are usually not
traffic-oblivious, and mostly focus on undirected graphs.
Concerning approaches on directed graphs, Foulds~\cite{Foulds1981273} minimizes
the cost of an \ac{MCF} in a bidirected network where the use of some
unidirectional arcs is prohibited to reduce congestion.
Gendron et al.~\cite{Gendron1999,DBLP:journals/ejco/GendronL14} discuss a
directed \ac{MCF} problem that considers costs for both the installation of an
edge and the amount of flow routed over it.
Further, in \emph{buy-at-bulk network
design}~\cite{DBLP:conf/waoa/Antonakopoulos10,DBLP:conf/soda/SalmanCRS97,DBLP:journals/siamcomp/ChekuriHKS10},
capacity on edges must be bought as cheaply as possible such that a given
traffic matrix becomes routable---with the caveat that larger amounts of
capacity can be bought at a lower price per capacity unit.

In \emph{robust network
design}~\cite{ben2005routing,DBLP:journals/tcs/Al-NajjarBL21,DBLP:journals/jcss/AzarCFKR04,DBLP:journals/talg/HajiaghayiKRL07},
possible traffic matrices are given as an uncertainty set in the form of a
polytope, and the objective is usually to minimize the cost of reserving
capacity on the edges.
While the \emph{dynamic routing} variant considers a different \ac{MCF} for
every traffic matrix in the polytope, \emph{static routing} specifies a fixed
unit flow for each commodity that is only scaled with the respective demand
value.
For directed graphs, Al-Najjar et al.~\cite{DBLP:conf/stacs/Al-NajjarBL22} show
that an exact algorithm for static routing
would yield an $\bigO(|V|)$-approximation for dynamic routing.
Our result can be seen as a better approximation ratio for the special uncertainty
set~$\{\mcfalpha\cdot\demand \ssep \mcfalpha \in (0,1),~\demand\text{ is
routable in }G\}$ under dynamic routing, but w.r.t.\ minimizing the number of
edges in a subgraph rather than the cost of reserved capacity.

There are also several related graph construction and subgraph minimization
problems:
Khuller et al.~\cite{DBLP:journals/ipl/KhullerRY94} give an approximation
algorithm for the construction of an undirected tree with constant degree that
accommodates given traffic demands between its leaves such that the maximum load
on any edge is minimized.
Otten et al.~\cite{DBLP:conf/lcn/OttenICA23} evaluate an \ac{ILP} and a
heuristic for a green traffic engineering problem on digraphs---however, there a
specific traffic matrix is also given %
(and rather than finding an edge subset of minimum
cardinality, they minimize the number of \enquote{line cards}, i.e., sets of 8
incident edges at each vertex).
Another well-known topic in the realm of subgraph minimization problems is that
of \emph{spanners}~\cite{DBLP:journals/csr/AhmedBSHJKS20}, i.e., subgraphs that
preserve the length of a shortest path within a given ratio (stretch factor)
between each pair of vertices.
There exists a correspondence between upper bounds on the stretch of
shortest paths and the congestion of \acp{MCF}, however, this only applies to the
existence of probabilistic mappings in undirected
graphs~\cite{DBLP:journals/corr/abs-0907-3631,DBLP:conf/stoc/Racke08}.
Nonetheless, this correspondence was used to find \emph{flow sparsifiers} in
undirected graphs~$G$~\cite{DBLP:journals/siamcomp/EnglertGKRTT14}.
While a \prob{MMCFS} solution is similar to a flow sparsifier that
preserves the congestion up to a factor~$\frac{1}{\mcfalpha}$, they differ in
that a flow sparsifier is an entirely new (undirected) graph, not necessarily subgraph,
containing a subset of the vertices of~$G$ but both old and new
edges~\cite{DBLP:conf/focs/Moitra09,DBLP:conf/stoc/LeightonM10,DBLP:conf/soda/AndoniGK14}.

Closely related to \prob{MMCFS} are classical \probf{DSND} problems, where, given
a (possibly capacitated) input digraph~$G=(V,E)$ and a requirement function~$r
\colon V^2 \to \N$, one aims to find an edge-wise minimum subgraph of~$G$ in
which one can send a flow of value~$r(s,t)$ from~$s$ to~$t$ for
every~$(s,t) \in V^2$.
On undirected graphs with unit edge capacities, there exists a 2-approximation by
Jain~\cite{DBLP:journals/combinatorica/Jain01},
which has been adapted to directed instances, but only for a very
restricted set of requirement
functions~\cite{DBLP:journals/networks/MelkonianT04}.
The \prob{DSND} problem most similar to \prob{MMCFS} is the
\probf{MCPS} problem~\cite{DBLP:conf/iwoca/ChimaniI23}, where the requirement value for a vertex
pair~$(s,t)$ equals a fraction~$\stretcha$ of the value of a maximum flow
from~$s$ to~$t$.
However, in all of the aforementioned \prob{DSND} approaches, each routed
commodity is considered in isolation from the others, whereas in the
\prob{MMCFS} setting, all commodities are routed simultaneously.
For example, given a digraph~$G=(V,E)$ with edge capacities~$\ecap$ and
a traffic matrix~$\demand$ routable in~$E$, the scaled-down matrix
$\mcfalpha\cdot \demand$ is not necessarily routable in any optimal
\prob{MCPS}-solution of~$(G,\ecap,\stretcha)$ since some edges may
be congested by the simultaneous routing of multiple commodities.
This is especially easy to see for $\mcfalpha \geq \stretcha$ but even holds true
when $\mcfalpha \ll \stretcha$:

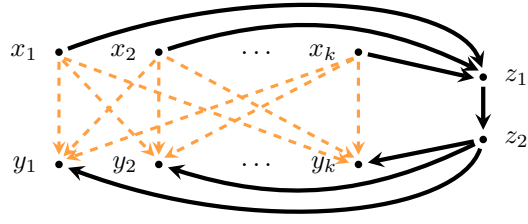
\begin{figure}[tbp]
    \centering
    \begin{tikzpicture}[scale=0.33]
        \node[label={[label distance=1mm]180:$x_1$},dot] (x1) at (0,0) {};
        \node[label={[label distance=1mm]180:$x_2$},dot] (x2) at (4,0) {};
        \node[label={[label distance=1mm]180:$x_k$},dot] (xn) at (12,0) {};
        \node[label={[label distance=1mm]180:$y_1$},dot] (y1) at (0,-4.5) {};
        \node[label={[label distance=1mm]180:$y_2$},dot] (y2) at (4,-4.5) {};
        \node[label={[label distance=1mm]180:$y_k$},dot] (yn) at (12,-4.5) {};
        \node[label={[label distance=1mm]0:$z_1$},dot] (x) at (17,-1) {};
        \node[label={[label distance=1mm]0:$z_2$},dot] (y) at (17,-3.5) {};

        \begin{scope}[every node/.style={circle,minimum size=1pt,inner sep=1pt}]
            \node (xdots) at (8,0) {$\dots$};
            \node (ydots) at (8,-4.5) {$\dots$};
        \end{scope}

        \begin{scope}[every path/.style={edge,ultra thick}]
            \path (x1) .. controls (8,3) and (16,2) .. (x.80);
            \path (y) .. controls (16,-6.5) and (8,-7.5) .. (y1);
        \end{scope}
        \begin{scope}[every edge/.style={edge,ultra thick}]
            \path (x2) edge[bend left=22] (x.85);
            \path (xn) edge (x.210);
            \path (y.-150) edge[bend left=22] (y2);
            \path (y) edge (yn);
            \path (x) edge[ultra thick]  (y);
        \end{scope}

        \begin{scope}[every edge/.style={edge, draw=myorange, very thick, dashed}]
            \path (x1.-90) edge (y1.90); \path (x1.-90) edge (y2.100); \path (x1.-90) edge (yn.180);
            \path (x2.-90) edge (y1.80); \path (x2.-90) edge (y2.90); \path (x2.-40) edge (yn.80);
            \path (xn.-90) edge (y1.60); \path (xn.-90) edge (y2.80); \path (xn.-90) edge (yn.90);
        \end{scope}
    \end{tikzpicture}
    \caption{Digraph~$G$ constructed in the proof of
        \Cref{th:mcf_cps_relationship}.
        Edges of the unique optimal \prob{MCPS}-solution~$E'$ of
        $(G,\ecap,\stretcha)$ are solid (black), and the remaining edges~$E_{XY}$
        are dashed (orange).}
    \label{fig:mcps_mcf_not_routable}
\end{figure}

\begin{observation}\label{th:mcf_cps_relationship}
    Given an arbitrarily small $\mcfalpha \in (0,1)$ and an arbitrarily large
    $\stretcha \in (0,1)$, there exists a digraph $G=(V,E)$ with edge capacities~$\ecap$
    and a traffic matrix~$\demand$ such that $\demand$ is routable
    in~$E$, but $\mcfalpha \cdot \demand$ is not routable in any
    optimal \prob{MCPS}-solution~$E'$ of~$(G,\ecap,\stretcha)$.
\end{observation}
\begin{proof}
    Let $C \be \ceil*{\frac{2\stretcha}{1-\stretcha}}$ be a (high) edge capacity
    value and $k > \sqrt{\frac{C}{\mcfalpha}}$ a number of vertices.
    Construct~$G=(V,E)$ (visualized in \Cref{fig:mcps_mcf_not_routable}) as
    follows:
    Create two vertex sets~$X$ and~$Y$ with $k$~vertices each, as well as
    two distinct vertices~$z_1, z_2$, and let $V \be X \union Y \union
    \{z_1,z_2\}$.
    Moreover, let~$E_{XY} \be X \times Y$, $E' \be (X \times \{z_1\}) \union
    \{z_1z_2\} \union (\{z_2\} \times Y)$, and $E \be E_{XY} \union E'$.
    Edge capacities and the traffic matrix are chosen as follows:
    \begin{align*}
        \ecap(uv) &= \begin{cases}
            1 & \text{if } uv \in E_{XY},\\
            C & \text{otherwise;}
        \end{cases}&
        \demand(u,v) &= \begin{cases}
            1 & \text{if } uv \in E_{XY},\\
            0 & \text{otherwise.}
        \end{cases}
    \end{align*}
    Every non-zero demand~$\demand(u,v)$ can be routed in $(G,\ecap)$ using the respective edge
    $uv \in E_{XY}$.

    The only optimal \prob{MCPS}-solution of~$(G,\ecap,\stretcha)$ is~$E'$:
    For every vertex pair $(u,v) \in E'$, the only $u$-$v$-path in~$G$ is the
    one consisting of the edge~$uv$---so the edges $E'$ must be in the solution.
    $E'$ also establishes a maximum flow of sufficient value for the remaining vertex pairs.
    In particular, for every vertex pair~$(u,v) \in X \times Y$, there exists a
    maximum $u$-$v$-flow of value~$C$ in $E'$, which is at least
    $\stretcha$~times the value $(C+1)$ of a maximum $u$-$v$-flow in $E$:
    \begin{align*}
        C = \ceil*{\frac{2\stretcha}{1-\stretcha}}
        \geq \stretcha\cdot\frac{2\cdot(1+\stretcha-\stretcha)}{1-\stretcha}
        = \stretcha \left(\frac{2\stretcha}{1-\stretcha} +2\right)
        \geq \stretcha \left(\ceil*{\frac{2\stretcha}{1-\stretcha}} +1\right)
        = \stretcha \cdot (C+1)
    \end{align*}
    However, the scaled matrix~$\mcfalpha \cdot \demand$ is not routable in $E'$:
    $|X|\cdot|Y| = k^2 > \frac{C}{\mcfalpha}$ many demands of value $\mcfalpha$
    would have to be routed over the single edge~$z_1z_2$, exceeding its
    capacity~$C$.
\end{proof}

Lastly, we want to highlight similarities of \prob{MMCFS} to the well-established
NP-hard \probf{MED}
problem~\cite{DBLP:journals/siamcomp/AhoGU72,DBLP:journals/siamcomp/Sahni74,DBLP:books/fm/GareyJ79,DBLP:conf/coco/Karp72},
where, given a digraph~$G$, one asks for the edge-wise minimum subgraph of~$G$
that preserves the reachability relation of~$G$.
In \Cref{sec:complexity}, we show that \prob{MED} is a special case of \prob{MMCFS}.
Further,
we observe:
\begin{observation}\label{th:med_in_mcfs}
    In a simple \ac{DAG}~$G$, the unique \ac{MED} of~$G$ must be contained in
    every feasible \prob{MMCFS} solution of~$G$, regardless of edge capacities
    and~$\mcfalpha$.
\end{observation}
\begin{proof}
    In any simple \ac{DAG}~$G$, the \ac{MED} is unique and consists of exactly
    those edges $st$ for which there is no $s$-$t$-path in $G -
    st$~\cite{DBLP:journals/siamcomp/AhoGU72}.
    A feasible \prob{MMCFS} solution must also contain these edges~$st$ in order to
    allow for a non-zero amount of flow from~$s$ to~$t$.
\end{proof}
However, in general digraphs, a feasible \prob{MMCFS} solution may not always
contain the \ac{MED}, see \Cref{fig:mcfs_may_not_contain_med}.
Note that \prob{MED} is not only polynomial-time solvable on \acp{DAG}, but
there are also several polynomial approximation algorithms for general
graphs~\cite{DBLP:journals/siamcomp/KhullerRF95,DBLP:journals/ipl/ZhaoNI03} with
the currently best approximation ratio
being~1.5~\cite{DBLP:conf/wads/BermanDK09,DBLP:conf/soda/Vetta01}.

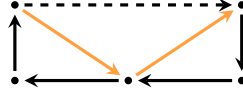
\begin{figure}[tbp]
    \centering
    \begin{tikzpicture}[scale=1]
            \begin{scope}[every node/.style={dot}]
                \node (a) at (0,1) {};
                \node (b) at (3,1) {};
                \node (c) at (0,0) {};
                \node (d) at (1.5,0) {};
                \node (e) at (3,0) {};
            \end{scope}

            \begin{scope}[every edge/.style={edge}]
                \path (a) edge[dashed] (b);
                \path (b) edge (e);
                \path (e) edge (d);
                \path (d) edge (c);
                \path (c) edge (a);
                \path (a) edge[myorange] (d);
                \path (d) edge[myorange] (b);
            \end{scope}
    \end{tikzpicture}
    \caption{An \prob{MMCFS} instance with an optimal solution (given by the
        solid edges) that does not contain the \ac{MED}. All edge capacities
        are~1, and $\mcfalpha = \frac{1}{2}$. The \ac{MED}, which is unique in
        this example and drawn in black, is not a feasible \prob{MMCFS} solution:
        for $\demand$ with $\demand(s,t) = \ecap(st)$ if $e=st \in E$
        and 0  otherwise, the dashed edge would have to accommodate a
        flow of~$\frac{3}{2}$ to satisfy all demands, but it only has a capacity of~1.}
    \label{fig:mcfs_may_not_contain_med}
\end{figure}

\section{Structural Results}\label{sec:structural}

We present some structural insights concerning \prob{MMCFS} that give a deeper
understanding of the problem.
Most importantly, we give a reformulation of \prob{MMCFS} that is used
throughout the rest of the paper to obtain structural and algorithmic results:
Recall that a feasible solution~$A$ for a given \prob{MMCFS} instance~$(G=(V,E),
\ecap, \mcfalpha)$ is defined as an edge set~$A \subseteq E$ such that for
\emph{all} traffic matrices~$T$ that are routable in~$E$,
$\mcfalpha\cdot\demand$ is routable in~$A$.
Interestingly, instead of explicitly considering all routable traffic
matrices~$T$, it suffices to consider the single specific traffic
matrix~$\alldemand$, which forces each edge to be utilized to its full capacity
but has no demands between non-adjacent vertices:
\begin{equation*}
    \alldemand(s,t) \be
    \begin{cases*}
        \ecap(e) & if $e=st \in E$, \\
        0 & otherwise.
    \end{cases*}
\end{equation*}
We show that an edge set~$A \subseteq E$ is a feasible \prob{MMCFS} solution iff
$\mcfalpha\cdot\alldemand$ is routable in $A$:

\begin{theorem}\label{th:super_t}
    Given a digraph~$G=(V,E)$ with edge capacities~$\ecap$, a retention
    ratio~$\mcfalpha\in(0,1)$, and an edge set~$A \subseteq E$, the following
    statements are equivalent:
    \begin{itemize}
        \item For all traffic matrices~$T$ that are routable in~$E$, the scaled
            matrix~$\mcfalpha\cdot\demand$ is routable in~$A$.
        \item The scaled matrix~$\mcfalpha\cdot\alldemand$ is routable in $A$.
    \end{itemize}
\end{theorem}
\begin{proof}
    $\alldemand$ is routable in $E$ by definition.
    If every traffic matrix routable in $E$ is also routable in~$A$ when scaled
    down by~$\mcfalpha$, then so is $\mcfalpha\cdot\alldemand$.
    For the other direction, consider any arbitrary traffic matrix~$T$ routable
    in~$E$.
    Let $\{\flow^{\demand}_{s,t} \ssep (s,t) \in V^2\}$ be the \ac{MCF} that
    routes~$\demand$ in~$E$ with the vector $\flowsum^\demand \be
    \sum_{(s,t) \in V^2} \flow^\demand_{s,t}$ specifying the total flow over
    each edge.
    Using this \ac{MCF}, we can construct a new traffic matrix~$\demand'$:
    \begin{equation*}
        \demand'(s,t) \be
        \begin{cases*}
            \flowsum^\demand(st) & if $e=st \in E$, \\
            0 & otherwise.
        \end{cases*}
    \end{equation*}

    Observe that $\demand'\leq\alldemand$ when using component-wise comparison.
    Thus, since $\mcfalpha\cdot\alldemand$ is routable in $A$, so is $\mcfalpha\cdot\demand'$.
    But if $\mcfalpha\cdot\demand'$ is routable in~$A$ using the
    flows~$\{\flow^{\mcfalpha\cdot\demand'}_{u,v} \ssep (u,v) \in V^2\}$, then
    $\mcfalpha\cdot\demand$ is also routable in~$A$ using the
    flows~$\{\flow^{\mcfalpha\cdot\demand}_{s,t} \ssep (s,t) \in V^2\}$
    constructed as follows:
    for each commodity~$st \in E$ and each edge~$uv\in E$, calculate the
    fraction of flow routed over $uv$ that is used by $\flow^\demand_{s,t}$, and
    route this fraction over the path chosen by
    $\flow^{\mcfalpha\cdot\demand'}_{u,v}$.
    In short,
    \begin{equation*}
        \flow^{\mcfalpha\cdot \demand}_{s,t}(e) \be
            \sum_{uv\in E \colon \flowsum^\demand(uv) > 0}
            \frac{\flow^\demand_{s,t}(uv)}{\flowsum^\demand(uv)}\cdot\flow^{\mcfalpha\cdot\demand'}_{u,v}(e). \qedhere
    \end{equation*}
\end{proof}

We note that a related but slightly different concept to
\Cref{th:super_t} has been implicitly used previously in
\cite{DBLP:conf/focs/Racke02} (on undirected graphs).
From this point onwards, we may refer to edges as commodities since
$\alldemand$ specifies a non-zero demand precisely for the edges in~$E$.
Moreover, given a flow~$\flow_{\comm}$ for a commodity~$\comm \in E$, we
call~$\flow_{\comm}(\comm)$ the \emph{direct flow} for~$\comm$.
We observe that in an optimal \prob{MMCFS} solution~$A$,
for every commodity~$\comm\in A$, the demand~$\alldemand(\comm)$ can always be
fully satisfied by a direct flow~$\flow_{\comm}(\comm)$:

\begin{observation}\label{th:direct_flow}
    Let~$G=(V,E)$ be a digraph with edge capacities~$\ecap$
    and~$\demand$ a traffic matrix routable in an edge set~$A \subseteq E$ with
    $\demand(e) \leq \ecap(e)$ for all edges $e \in E$.
    Then, there exists an \ac{MCF}~$\flowset = \{\flow_{s,t} \ssep (s,t) \in
    V^2\}$ in the graph~$G'\be(V,A)$ satisfying the demands~$\demand$ such that
    $\flow_{\comm}(\comm) = \demand(\comm)$ for all edges $\comm\in A$.
\end{observation}
\begin{proof}
    Among all \acp{MCF} that witness the routability of $T$ in $A$,
    let~$\flowset'=\{\flow'_{s,t} \ssep (s,t) \in V^2\}$ be one with a maximum
    sum of direct flow values~$\sum_{\comm \in E} \flow'_{\comm}(\comm)$.

    We give a proof by contradiction:
    Assume that there exists an edge~$e'=uv \in A$ such that $\flow'_{e'}(e') <
    \demand(e')$ (if no such edge exists, $\flowset = \flowset'$ and we are
    done).
    There must exist at least one alternative $u$-$v$-path~$P$ that routes at
    least some of the remaining demand~$\demand(e') - \flow'_{e'}(e')$.
    Further, the edge~$e'$ has residual capacity $\ecap(e') -
    \sum_{(s,t)\in V^2 : }\flow'_{s,t}(e') = 0$ as otherwise we could increase
    $\flow_{e'}(e')$ (and decrease flow along $P$ accordingly), which would
    contradict the selection of~$\flowset'$.
    Thus, there exists an edge~$e'' \in E$, $e'' \neq e'$, with $\flow_{e''}(e')
    > 0$.
    But then, we can exchange a non-zero amount~$\varepsilon > 0$ of flow of
    commodity~$e'$ routed over~$P$ with an equally small amount of flow of
    commodity~$e''$ routed over~$e'$.
    This increases $\flow'_{e'}(e')$ without decreasing any other direct flow
    value---again a contradiction to the selection of $\flowset'$.
\end{proof}

Further, for any edge set~$A$, we can compare its total edge capacity~$\sum_{e
\in A} \ecap(e)$ to the total flow needed to satisfy the demands~$\alldemand$.
This not only gives us a necessary condition for an edge set~$A$ to be a
feasible \prob{MMCFS} solution, but, upon closer analysis, also allows us to
relate its quality as a solution to its \emph{mean capacity}~$\meancap{A} \be
\frac{1}{|A|}\cdot\sum_{e \in A} \ecap(e)$.
The following results apply both in the case of simple and non-simple graphs,
but we can give better guarantees in the former case.

\begin{theorem}
    Let~$O$ be an optimal solution and $A\neq O$~a feasible solution for an
    \prob{MMCFS} instance~$(G,\ecap,\mcfalpha)$.
    Then, $\frac{|A|}{|O|} \leq \min\left\{
    \left(1+\frac{1-\mcfalpha}{\theta\mcfalpha}\right) \cdot
    \frac{\meancap{O}}{\ \meancap{A}\ }, 1 + \frac{1-\mcfalpha}{\theta\mcfalpha}
\cdot\frac{\meancap{O}}{\ \meancap{A\setminus O}\ }\right\}$
    with $\theta = 2$ if $G$ has no parallel edges and $\theta = 1$ otherwise.
\end{theorem}
\begin{proof}
    Let $X \be A \setminus O$, and $Y \be A \intersect O$.
    The commodities of all~$\comm \in O$ must be routed through~$O$, requiring a
    total flow of at least~$\mcfalpha \sum_{e \in O} \ecap(e)$ and leaving a
    total remaining capacity in~$O$ of at most~$(1-\mcfalpha) \sum_{e \in O}
    \ecap(e)$.
    Every commodity~$\comm' \in X$ has to be routed within this remaining
    capacity since~$O$ is feasible.
    This requires a total flow of at least $\theta \mcfalpha \sum_{e \in X}
    \ecap(e)$; the $\theta$ is due to the fact that without parallel edges each such
    commodity $\comm'$  must be routed over at least two other edges in~$O$
    since $\comm'\not\in O$.
    We thus have
    \begin{equation}\label{eq:capacity_relation}
        \theta \mcfalpha \sum_{e \in X} \ecap(e) \leq (1-\mcfalpha) \sum_{e \in
         O} \ecap(e)~.
    \end{equation}
    By adding $\theta \mcfalpha \sum_{e \in Y} \ecap(e) \leq \theta \mcfalpha
    \sum_{e \in O} \ecap(e)$ to this inequality, we obtain
    \begin{align*}
        \theta \mcfalpha \sum_{e \in A} \ecap(e)
            &\leq (1-\mcfalpha) \sum_{e \in O} \ecap(e) + \theta \mcfalpha \sum_{e \in O} \ecap(e)\\ %
        \theta \mcfalpha \cdot |A|\cdot\meancap{A}
            &\leq (1-\mcfalpha+\theta\mcfalpha) \cdot|O|\cdot\meancap{O}\\
        \frac{|A|}{|O|}
            &\leq \left(1+\frac{1-\mcfalpha}{\theta\mcfalpha}\right) \cdot
            \frac{\meancap{O}}{\ \meancap{A}\ }~.
    \end{align*}
    Alternatively, we can rewrite inequality~\eqref{eq:capacity_relation}
    as~$\theta \mcfalpha \cdot |X|\cdot\meancap{X} \leq (1-\mcfalpha)
    \cdot|O|\cdot\meancap{O}$ and obtain
    \begin{align*}
            \frac{|A|}{|O|} \leq \frac{|O|+|X|}{|O|} = 1 + \frac{|X|}{|O|}
            &\leq 1 + \frac{1-\mcfalpha}{\theta\mcfalpha}
            \cdot\frac{\meancap{O}}{\ \meancap{X}\ }~. \qedhere
    \end{align*}
\end{proof}

\begin{corollary}\label{th:mcfs_trivial_approx}
    Any arbitrary feasible solution for \prob{MMCFS} (including the trivial one,
    $E$ itself), is a $(1+\frac{1-\mcfalpha}{\theta \mcfalpha} \cdot
    \frac{\max_{e \in E}\ecap(e)}{\min_{e \in E}\ecap(e)})$-approximation.
\end{corollary}

This ratio is met, e.g., on a bundle of parallel edges with
$\ell \be (\frac{1}{\mcfalpha} -1)\cdot k$ capacity-1 edges and one capacity-$k$ edge
(for any given $\mcfalpha\in\Q_{(0,1)}$ and an arbitrary $k\in\Q$
s.t.\ $\ell \in\N$).

\begin{corollary}\label{th:mcfs_unit_cap_approx}
    For uniform capacities, any arbitrary feasible solution for \prob{MMCFS}
    (including the trivial one, $E$ itself) is a $(1+\frac{1-\mcfalpha}{\theta
    \mcfalpha})$-approximation.
\end{corollary}

\section{Complexity}\label{sec:complexity}

Given that even trivial \prob{MMCFS} solutions satisfy an approximation guarantee
according to \Cref{th:mcfs_trivial_approx}, one might expect \prob{MMCFS} to be
polynomial-time solvable.
However, in this section, we show that \prob{MMCFS} is NP-hard already on
\acp{DAG} and give a first inapproximability result.
We begin by proving that \prob{MMCFS} is NP-hard already with unit edge capacities
using a reduction from \prob{MED} that directly follows from \Cref{th:super_t}:

\begin{corollary}
    \prob{MED} is the special case of \prob{MMCFS} with unit edge
    capacities~$\ecap$ and $\mcfalpha \leq \frac{1}{|E|}$.
\end{corollary}
\begin{proof}
    An optimal solution~$A\subseteq E$ for an \prob{MMCFS}
    instance~$(G,\ecap,\mcfalpha)$ with unit edge capacities is an edge set of
    minimum cardinality such that the demands~$\mcfalpha\cdot\alldemand(s,t)
    =\mcfalpha\cdot\ecap(st) \leq \frac{1}{|E|}$ for each edge $st\in E$
    are routable in~$A$.
    This is equivalent to ensuring that there exists an $s$-$t$-path in~$A$ for
    every edge~$st \in E$
    since the (unit) capacity of an edge can never be surpassed by the~$|E|$
    many flows of size at most~$\frac{1}{|E|}$ each.
\end{proof}

Moreover, we can show that \prob{MMCFS} is NP-hard already on \acp{DAG} using a
reduction from the NP-hard decision variant of \probl{SC}~\cite{DBLP:conf/coco/Karp72},
where one asks:
given a universe~$\univ$, a family of sets $\sets = \{S_i \subseteq \univ
\}_{1\leq i\leq k}$ with $k \in \bigO(\text{poly}(|\univ|))$, and a
parameter~$\varphi$, is there a subfamily $\setcover \subseteq \sets$ of
cardinality $|\setcover| \leq \varphi$ such that $\Union_{S \in \setcover} S =
\univ$?
The reduction is similar to that given in~\cite{DBLP:conf/iwoca/ChimaniI23} to prove the
NP-hardness of the \probl{MCPS} problem.

\newcommand{\Emed}{E_1}
\newcommand{\Esets}{E_{\sets}}
\newcommand{\Eitems}{E_{\univ}}
\begin{theorem}
    For any fixed~$\mcfalpha\in(0,1)$, \prob{MMCFS} is NP-hard already on
    \acp{DAG} where the longest path has length~3.
\end{theorem}
\begin{proof}
    Given a \probl{SC} instance~$(\univ,\sets,\varphi)$ and a fixed retention
    ratio $\mcfalpha\in(0,1)$, we construct an
    instance~$I=(G=(V,E),\ecap,\mcfalpha,\psi)$ for the decision variant of
    \prob{MMCFS}: $I$~is a yes-instance if and only if there exists a feasible
    \prob{MMCFS} solution~$E'\subseteq E$ for~$(G,\ecap,\mcfalpha)$ with
    cardinality~$|E'| \leq \psi$.
    We construct~$I$ as follows (see \Cref{fig:set_cover_to_mcfs_reduction}
    for a visualization):
    \begin{align*}
        V&\be \mathrlap{V_\univ \union V_\sets \union V^\sets_t \union \{t\}
            \text{ and } E\be \Eitems \union E_\sets \union \Emed} &&&&\\
        V_\univ &\be \{v_{u},v'_{u} \mid \forall  u \in \univ\} &
        V_\sets &\be \{v_S \mid \forall S \in \sets\} &
        V^t_\sets &\be \{z_S \mid \forall S \in \sets\}\\
        \Eitems &\be V_U \times \{t\} &
        \Esets &\be V_\sets \times \{t\} & &\\%
        \Emed &\be \mathrlap{\{v_uv_S, v'_uv_S\mid \forall S \in \sets, u \in S\} \union \{v_Sz_S,z_St \mid s \in \sets\}} \\
        \ecap(e) &\be \mathrlap{\left\{%
        1~\text{if}~e \in \Emed;\qquad
        {\textstyle \frac{1-\mcfalpha}{\mcfalpha}}~\text{if}~e \in \Esets;\qquad
        \varepsilon~\text{if}~e \in \Eitems,~\text{with}~\varepsilon \leq \min\left({\textstyle \frac{1-\alpha}{\alpha^2 \cdot |\univ|}},
        {\textstyle \frac{1-\alpha}{\alpha}}\right)
        .\right.}\\
        \psi &\be \rlap{$|\Emed| + \varphi = 2\cdot\sum_{S \in \sets} |S| + 2\cdot|\sets| + \varphi$}
    \end{align*}
    As $G$ is a \ac{DAG}, its \ac{MED} is
    unique~\cite{DBLP:journals/siamcomp/AhoGU72} and must be part of any
    feasible \prob{MMCFS} solution, see \Cref{th:med_in_mcfs}.
    This \ac{MED} is formed by the edges~$e \in \Emed$; a flow
    of~$\mcfalpha$ must be routed over each of them in order to
    satisfy the demands~$\alldemand(e) = \mcfalpha\cdot\ecap(e)$.
    The remaining capacity for each edge~$e \in \Emed$
    is~$1-\mcfalpha$.

    So consider a single set~$S\in\sets$ and the corresponding
    two-path~$\{v_Sz_S,z_St\} \in \Emed$, whose remaining capacity can thus
    accommodate either arbitrarily many \emph{item
    commodities}~$\comm_{\univ} \in \Eitems$ (each one with the sufficiently small
    demand~$\alldemand(\comm_{\univ}) = \mcfalpha\cdot\varepsilon$) or a
    single \emph{set commodity}~$\comm_{\sets} \in \Esets$ (with the
    demand~$\mcfalpha \cdot \alldemand(\comm_{\sets}) = \mcfalpha \cdot
    \frac{1-\mcfalpha}{\mcfalpha} = 1-\mcfalpha$).
    In the former case, we can remove at least two \emph{corresponding item
    edges}~$v_ut$, $v'_ut$ with $u \in S$, which is more than the single
    \emph{corresponding set edge}~$v_St$ we can remove in the latter case.
    Thus, for each item commodity $v_ut \in \Eitems$, an optimal \prob{MMCFS}
    solution must contain one of the corresponding set edges~$\{v_St \ssep S \ni
    u\} \subseteq \Esets$; the item commodity can then be routed over the path
    from~$v_u$ over~$v_S$ to~$t$.

    Given a \probl{SC} solution $\setcover$ with $|\setcover| \leq \varphi$, we can
    construct an \prob{MMCFS} solution~$E' = \Emed \union \{v_St
    \in \Esets \ssep S \in \setcover\}$ with cardinality
    $|E'| = |\Emed| + |\setcover| \leq
    |\Emed| + \varphi = \psi$.
    Since every item is covered by the sets in $\setcover$, the constructed
    \prob{MMCFS} solution includes at least one corresponding set edge for
    each item, ensuring its feasibility.
    Conversely, since a feasible \prob{MMCFS} solution~$E'$ with $|E'| \leq \psi$
    has at least one corresponding set edge for each item $u \in \univ$, the
    \probl{SC} solution $\setcover = \{S \ssep v_St \in \Esets \intersect E'\}$ also
    contains at least one covering set for each item.
    Moreover, $|\setcover| = |E'| - |\Emed| \leq \psi - |\Emed| = \varphi$.
\end{proof}

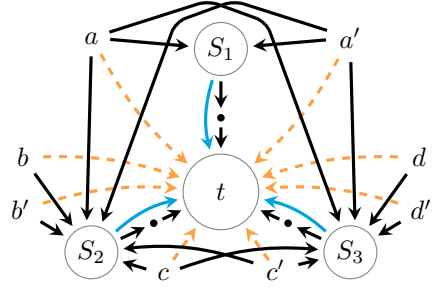
\begin{figure}[t]
\centering
\begin{tikzpicture}[scale=0.7]
        \begin{scope}[every node/.style={circle,minimum size=1pt,inner sep=1pt}]
            \node[minimum size=1cm,draw=gray] (t) at (0,0) {$t$};
            \node[dot] (s1_1) at (90:1.4cm) {};
            \node[dot] (s2_1) at (205:1.4cm) {};
            \node[dot] (s3_1) at (335:1.4cm) {};
            \node[minimum size=0.7cm,draw=gray] (s1) at (90:2.7cm) {$S_1$};
            \node[minimum size=0.7cm,draw=gray] (s2) at (205:2.7cm) {$S_2$};
            \node[minimum size=0.7cm,draw=gray] (s3) at (335:2.7cm) {$S_3$};
            \node (a)  at (130:3.8cm) {$a$};
            \node (b)  at (305:1.8cm) {$c'$};
            \node (c)  at (185:3.8cm) {$b'$};
            \node (d)  at (355:3.8cm) {$d'$};
            \node (a2) at  (50:3.8cm) {$a'$};
            \node (b2) at (235:1.9cm) {$c$};
            \node (c2) at (170:3.8cm) {$b$};
            \node (d2) at  (10:3.8cm) {$d$};
        \end{scope}

        \begin{scope}[every path/.style={edge, draw=myblue, very thick,bend right=15}]
            \path (s1) edge (t);
            \path (s2) edge[bend left=15] (t);
            \path (s3) edge (t);
        \end{scope}

        \begin{scope}[every path/.style={edge, draw=myorange, dashed, very thick,rounded corners}]
            \path[] (a) edge[bend right=10] (t);
            \path[] (b) edge[shorten <= -0.08cm] (t);
            \path[] (c) edge[bend left=10] (t);
            \path[] (d) edge[bend right=10] (t);
            \path[] (a2) edge[bend left=10] (t);
            \path[] (b2) edge (t);
            \path[] (c2) edge[bend left=10] (t);
            \path[] (d2) edge[bend right=10] (t);
        \end{scope}

        \begin{scope}[every path/.style={edge}]
            \draw (a) -- (s1); \draw (a2) -- (s1);
            \draw (a) -- (s2.100); \draw (a2) -- ($(a2)+(-2.2,0.7)$) -- ($(a2)+(-3.3,0.3)$) -- (s2);
            \draw (a) -- ($(a)+(2.2,0.7)$) -- ($(a)+(3.3,0.3)$) -- (s3); \draw (a2) -- (s3.80);
            \path (b) edge[bend right=12] (s2); \draw (b2) edge (s2);
            \path (b) edge (s3); \path (b2) edge[bend left=12] (s3);
            \draw (c) -- (s2); \draw (c2) -- (s2);
            \draw (d) -- (s3); \draw (d2) -- (s3);
            \draw (s1) -- (s1_1); \draw (s1_1) -- (t);
            \draw (s2) -- (s2_1); \draw (s2_1) -- (t);
            \draw (s3) -- (s3_1); \draw (s3_1) -- (t);
        \end{scope}
\end{tikzpicture}
\caption{\prob{MMCFS} instance constructed from a \probl{SC} instance
    with universe~$\univ = \{a,b,c,d\}$ and family of subsets~$\sets=\{\{a\},\
    \{a,b,c\},\ \{a,c,d\}\}$.
    An optimal solution contains the (solid black) \ac{MED} as well as
    one (blue) corresponding set edge for each~$u \in
    \univ$. Item edges are orange and~dashed.}
    \label{fig:set_cover_to_mcfs_reduction}
\end{figure}

\newcommand{\numaltered}{\mu}
Considering the optimization variants of \probl{SC} and \prob{MMCFS} (and thus
ignoring the additional input values $\varphi$ and~$\psi$),
the reduction above also implies the inapproximability of the number
of edges in an optimal \prob{MMCFS} solution beyond the edges required for an
\ac{MED}:
Consider an instance~$I = (G=(V,E),\ecap,\mcfalpha)$ for the optimization
variant of \prob{MMCFS} that is produced by the reduction above, and an arbitrary
feasible solution~$A \subseteq E$ for this instance.
Let $\numaltered(I,A) \be |A| - \medval(G)$ where $\medval(G)$ denotes the
number of edges in an \ac{MED} of~$G$.
Then, $A$ can be transformed into a feasible solution for the original \probl{SC}
instance with objective value $\numaltered(I,A)$ in linear time.
Further, let~$\numaltered(I)$ be the minimum~$\numaltered(I, A')$ over all
feasible solutions~$A'$ for $I$, and recall that the size~$|E|$ of the
\prob{MMCFS} instance~$I$ is linear in the size~$N \in \bigO(\text{poly}(\univ))$
of the \probl{SC} instance:
if it was possible to approximate $\numaltered(I)$ within a factor
in $o(\log|E|) = o(\log|U|)$, one
could also approximate \probl{SC} within $o(\log|U|)$, which is
NP-hard~\cite{DBLP:conf/stoc/DinurS14,DBLP:journals/toc/Moshkovitz15}.
This implies that any approximation algorithm for \prob{MMCFS} (such as the one
we present in \Cref{sec:mcfs_approx}) must leverage the existence of a
comparatively high number of \ac{MED}-edges in order to achieve its
approximation ratio.

\begin{observation}
    Given an \prob{MMCFS} instance $I = (G=(V,E),\ecap,\mcfalpha)$ and an
    \ac{MED} of $G$, approximating $\numaltered(I)$ with a ratio in $o(\log|E|)$
    is NP-hard.
    This already holds on \acp{DAG} where the longest path has length~3.
\end{observation}

\section{LP-based Approximation}\label{sec:mcfs_approx}

We present an extremely simple
$\max(\nicefrac{1}{\mcfalpha}, 2)$-approximation for \prob{MMCFS} based on LP
rounding.
This is a clear improvement over the default approximation guarantee of
\Cref{th:mcfs_trivial_approx}, which depends the instance's edge capacities and
is thus not even polynomially bounded in~$\mcfalpha^{-1}$ or the instance size.
Consider the following ILP formulation for \prob{MMCFS}:
\begin{ilp}{1}\label{eq:ilp_mcfs}
\begin{align}
    \min &\sum_{e \in E}x_e & \label{eq:ilp_obj}\\
    \sum_{u\colon uv \in E} \flow_{\comm}(uv) - \sum_{u \colon vu \in E} \flow_{\comm}(vu) &=
            \begin{cases*}
                -\mcfalpha\alldemand(\comm) & if $v=s$\\
                \mcfalpha\alldemand(\comm) & if $v=t$\\
                0 & else
            \end{cases*}
        &\forall v\in V, \comm\!=\!st \in E \label{eq:ilp_flow_preservation}\\
    \sum_{\comm \in E} \flow_{\comm}(e) &\leq x_e \cdot \ecap(e)
                       &\forall e\in E \label{eq:ilp_cap_constraint}\\
    \flow_{\comm}(e) &\geq 0 &\forall e \in E, \comm \in E\\
    x_e & \in \{0,1\} &\forall e \in E \label{eq:ilp_integrality_constraint}
\end{align}
\end{ilp}

A binary indicator variable~$\lpfeas_e$ determines whether edge~$e\in E$ is part
of the solution subgraph or not.
We minimize the sum of these variables in the objective
function~\eqref{eq:ilp_obj} to obtain a subgraph of minimum size.
The non-negative variables $\flow_{\comm}(e)$ determine the amount of flow routed
over edge~$e$ for commodity $\comm \in E$.
Recall that $\alldemand$ specifies a demand of $\ecap(\comm)$ precisely for each
edge~$\comm$:
the flow preservation constraints~\eqref{eq:ilp_flow_preservation} guarantee
that the $\flow$-variables represent proper flows that satisfy these demands.
Lastly, the capacity constraints~\eqref{eq:ilp_cap_constraint} ensure that the
total sum of flow over any edge~$e \in E$ does not surpass the capacity
of~$e$.
This flow sum~$\flowsum(e) \be \sum_{\comm \in E}
\flow_{\comm}(e)$ must be zero if~$e$ is not part of the solution.

We obtain the relaxation of ILP~\eqref{eq:ilp_mcfs} by replacing the integrality
constraints~\eqref{eq:ilp_integrality_constraint} on $x_e$ by the inequalities
$0 \leq x_e \leq 1$ for all $e \in E$.
The only other type of constraint that bounds the $x_e$-variables
is~\eqref{eq:ilp_cap_constraint}.
Since the $x_e$-variables can assume fractional values in the LP relaxation, and
the sum over all $x_e$ is minimized, this lower bound
of~$\frac{\flowsum(e)}{\ecap(e)}$ on~$x_e$ for all $e \in E$ will always be
met with equality.
Hence, the LP relaxation is equivalent to a standard \ac{MCF}-LP that
minimizes the sum of edge utilizations in the objective function~$\sum_{e \in E}
\frac{\flowsum(e)}{\ecap(e)}$.
Accordingly, we will refer to $\cost(e) \be \frac{1}{\ecap(e)}$ as the
\emph{cost} that routing a single unit of flow over an edge~$e$ will add to this
objective.

\begin{remark}
    The solution~$\lpfeas_e = \mcfalpha$ for all~$e \in E$ with objective
    value~$\mcfalpha|E|$ is always feasible~for the relaxation of
    ILP~\eqref{eq:ilp_mcfs}.
    When the input graph is simple and all edge capacities are uniform, it is in fact
    optimal: the flow~$\flow_{\comm}$ for commodity~$\comm$ will always be
    routed completely over~$\comm$ since all edges have the same cost and routing
    $\flow_{\comm}$ over an alternative path would cost more than routing it over a
    single edge.
    This implies that any arbitrary feasible solution for \prob{MMCFS} on simple
    graphs with uniform edge capacities is an $(\nicefrac{1}{\mcfalpha})$-approximation
    (which we already proved using a separate argument via
    \Cref{th:mcfs_unit_cap_approx}).
    However, below we also use the LP relaxation to
    approximate \prob{MMCFS} on non-simple graphs with non-uniform edge capacities.
\end{remark}

For our approximation algorithm, we make use of the well-known fact that all
standard LP solving algorithms
will always return a \emph{basic} optimal solution (if any solution
exists)~\cite[p.\ 279]{DBLP:books/daglib/0030297}.
A basic optimal solution---or equivalently, extreme point solution---is a vertex
of the polyhedron defined as the convex hull of the set of feasible solutions.
It does not lie on a higher-dimensional face of the polyhedron and thus cannot
be expressed as a convex combination of two or more other feasible solutions
\cite[p.\ 100]{DBLP:books/daglib/0004338}.
We prove that a basic optimal solution for the relaxation of
ILP~\eqref{eq:ilp_mcfs} will only have comparatively few variables with a
positive value below $\mcfalpha$, allowing us to round up the solution
to obtain an approximation.

\newcommand{\numsat}{\ensuremath{h}}
\newcommand{\numtiny}{\ensuremath{\ell}}
\newcommand{\satedges}{\ensuremath{H}}
\newcommand{\tinyedges}{\ensuremath{L}}
\begin{lemma}\label{th:lp_relaxation_bound}
    Let~$\lpopt$ be a basic optimal solution for the relaxation of
    ILP~\eqref{eq:ilp_mcfs},
    and $\numsat$~the number of edges~$e$ such that $\lpopt_e = 1$.
    There exist at most $\numsat$ many edges~$e'$ with $\lpopt_{e'} \in (0, \mcfalpha)$.
\end{lemma}
\begin{proof}
    Let
    $\satedges \be \{e \in E \ssep \lpopt_e = 1\}$ be the $\numsat$ many edges
    that are saturated by the fractional \acl{MCF}, and $\tinyedges \be \{e \in E
    \ssep \lpopt_e \in (0,\mcfalpha)\}$ with $\numtiny \be |\tinyedges|$ the
    edges that are only used to a fraction less than $\mcfalpha$;
    in short, edges with \emph{high} and \emph{low} corresponding~$\lpopt$-values.
    We show that an optimal fractional solution~$\lpopt$ with $\numtiny >
    \numsat$ would allow us to construct a vector~$p \in \Q^{|E|}$ such that
    both $(\lpopt + p)$ and $(\lpopt - p)$ are still feasible
    solutions for the LP relaxation---thus, $\lpopt$~would not be basic.

    \newcommand{\tmpvec}{q}
    For every edge $e=st \in \tinyedges$, let $P_e$ denote an
    arbitrary \emph{alternative $s$-$t$-path} not using the edge~$e$
    and $\flow_e(e') > 0$ for all edges $e' \in P_e$.
    Such a path must exist since at most $\lpopt_e \cdot \ecap(e) < \mcfalpha \cdot
    \ecap(e)$ flow is routed over~$e$, so an alternative $s$-$t$-path is
    necessary to satisfy the demand $\mcfalpha\cdot\alldemand(s,t) = \mcfalpha
    \cdot \ecap(e)$.
    Further, the total cost of routing a unit of flow over~$P_e$ must be
    lower than or equal to the cost of routing it over $e$ itself, i.e.,
    $\sum_{e' \in P_e} \frac{1}{\ecap(e')} \leq \frac{1}{\ecap(e)}$, by
    the optimality of $\lpopt$.
    Based on the alternative paths, we construct a matrix~$M \in
    \{0,1\}^{\numtiny\times\numsat}$ indexed by pairs $(e,e'') \in \tinyedges
    \times \satedges$, with
    $M(e,e'') \be 1~\text{if}~e'' \in P_e,~\text{and}~0~\text{otherwise}$.
    Since $\numtiny > \numsat$, the $\numtiny$~rows of $M$ must be linearly
    dependent, i.e., there exists a vector of coefficients $\tmpvec \in
    \Q^{\numtiny}$ with $\tmpvec \neq \mathbf{0}$, such that $\transpose{\tmpvec} \cdot M =
    \mathbf{0}$ (and consequently, $-\transpose{\tmpvec} \cdot M = \mathbf{0}$).

    We can obtain two new feasible solutions by modifying the \ac{MCF}
    corresponding to the optimal basic solution $\lpopt$ as follows:
    for each edge~$e \in \tinyedges$, and using a small positive
    value~$\varepsilon \in \Q$, we send $\varepsilon \cdot \tmpvec(e)$ less flow
    over~$e$ itself and $\varepsilon \cdot \tmpvec(e)$ more flow over
    the alternative path $P_e$---or vice versa.
    Put formally, we argue that for a small enough~$\varepsilon > 0$, the
    following vector~$p \in \Q^{|E|}$ yields two feasible solutions
    $(\lpopt + p)$ and $(\lpopt - p)$ for the LP relaxation:
    \begin{align*}
        p(e') =
        \begin{cases*}
            \varepsilon\cdot\tmpvec(e') -\varepsilon \cdot \sum_{e \in \tinyedges \colon e' \in P_e} \tmpvec(e) & if $e' \in \tinyedges$,\\
            -\varepsilon\cdot \sum_{e \in \tinyedges \colon e' \in P_e}\tmpvec(e) & otherwise. %
        \end{cases*}
    \end{align*}

    Clearly, the two solutions satisfy all demands and flow preservation
    constraints~\eqref{eq:ilp_flow_preservation}.
    Consider the capacity constraints~\eqref{eq:ilp_cap_constraint}:
    By construction of~$\tmpvec$, the flow difference on saturated edges is 0.
    For all non-saturated edges $e'$ it holds: if we modified the flow routed over~$e'$,
    this flow was already non-zero before our modification, and $\varepsilon$
    can be chosen sufficiently small such that the flow will neither turn
    negative nor surpass~$\ecap(e')$.%

    \newcommand{\emax}{e_{\max}}
    It remains to show that $p \neq \mathbf{0}$ given that $\tmpvec \neq
    \mathbf{0}$.
    So, among the edges~$e \in \tinyedges$ with $\tmpvec(e) \neq 0$, choose
    one with maximum cost~$\frac{1}{\ecap(e)}$ and denote it by~$\emax$.
    We show that $p(\emax) \neq 0$.

    If~$\emax$ were contained in any of the alternative paths~$P_e$ with $e \in
    \tinyedges$ and $\tmpvec(e) \neq 0$, we would arrive at a
    contradiction even without using $p$:
    Recall that $\sum_{e' \in P_e} \frac{1}{\ecap(e')} \leq
    \frac{1}{\ecap(e)}$.
    So we must have $|P_e| = 1$, i.e., $e$ and $\emax$ are parallel
    edges with $\ecap(e) = \ecap(\emax)$.
    Since $e,\emax\in L$ and thus $\lpopt_e, \lpopt_{\emax} \in (0,\mcfalpha)$,
    we could simply obtain two new solutions by shifting some small
    $\varepsilon >0$ of flow from one to another, or vice versa, contradicting
    that $\lpopt$ is basic.

    Hence, $\emax$ is only contained in alternative paths~$P_e$ with~$e \in
    \tinyedges$ s.t.\ $\tmpvec(e) = 0$, and
    \begin{align*}
        p(\emax)
        = \varepsilon\cdot\tmpvec(\emax) -\varepsilon \cdot \sum_{e \in \tinyedges \colon \emax \in P_e} \tmpvec(e)
        = \varepsilon\cdot\tmpvec(\emax) - 0 \neq 0.&\qedhere
    \end{align*}
\end{proof}

The proof implicitly gives an intuitive explanation for the distribution of LP
values:

\begin{corollary}\label{th:saturated_edges}
    Let~$\lpopt$ be a basic optimal solution for the relaxation of
    ILP~\eqref{eq:ilp_mcfs}.
    For each edge $e=st$ with $\lpopt_e \in (0,\mcfalpha)$ it holds:
    every \emph{alternative $s$-$t$-path}~$P_e$ (not containing~$e$)
    with~$\flow_e(e') > 0$ for all of its edges $e' \in P_e$ must contain an
    edge~$e''$ with $\lpopt_{e''} = 1$.
\end{corollary}
\begin{proof}
    Assume that, for some edge~$e=st$ with~$\lpopt_e \in (0,\mcfalpha)$, there
    is a~$P_e$ that does not contain an edge~$e''$ with $\lpopt_{e''} = 1$.
    Then, we can follow the proof of \Cref{th:lp_relaxation_bound},
    choosing~$P_e$ as the alternative $s$-$t$-path for~$e$.
    As a result, the matrix~$M$ constructed in the proof contains a row of
    zeroes, allowing us to route~$\varepsilon$ more flow over~$e$
    and~$\varepsilon$ less flow over $P_e$, or vice versa.
    Thus, $\lpopt$ cannot be basic, a contradiction.
\end{proof}

We now analyze the following LP-rounding algorithm: compute a basic optimal
solution~$\lpopt$ for the relaxation of ILP~\eqref{eq:ilp_mcfs} in polynomial
time, and return the edge set~$\{e \in E \ssep \lpopt_e > 0\}$.
Based on \Cref{th:lp_relaxation_bound}, one could use na\"ive techniques to prove
approximation guarantees of~$(2 + \frac{1}{\mcfalpha})$ or~$\frac{2}{\mcfalpha}$
for this algorithm.
However, we provide the stronger bound of~$\max(\nicefrac{1}{\mcfalpha}, 2)$
based on the following intuition:
the algorithm either \enquote{misses} the optimal solution mainly due to edges
with $\lpopt_e \in (0,\mcfalpha)$ and we can obtain a 2-approximation via
\Cref{th:lp_relaxation_bound},
or due to edges with $\lpopt_{e} \in [\mcfalpha,1)$ and we can
round up the variables for a $\frac{1}{\mcfalpha}$-approximation.

\begin{theorem}\label{th:mcfs_approx}
    Let~$\lpopt$ be a basic optimal solution for the relaxation of ILP~\eqref{eq:ilp_mcfs}.
    The edge set $A \be \{e \in E \ssep \lpopt_e > 0\}$ is a
    $\max(\nicefrac{1}{\mcfalpha}, 2)$-approximation for \prob{MMCFS}.
    That is, rounding up~$\lpopt$ is a 2-approximation when $\mcfalpha \geq
    \frac{1}{2}$, and an $\nicefrac{1}{\mcfalpha}$-approximation when $\mcfalpha < \frac{1}{2}$.
\end{theorem}
\begin{proof}
    The solution~$A$ is clearly feasible.
    Let $\algval \be |A|$ be the number of edges~$e$ such that $\lpopt_e > 0$.
    Furthermore, let
    $\Delta = |\{e \in E  \ssep \lpopt_e \in \mathopen[\mcfalpha, 1)\}|$
    be the number of these $\algval$ edges whose $\lpopt$-variable is set to a value
    in the interval $\mathopen[\mcfalpha, 1)$.
    We know that the remaining $(\algval-\Delta)$ edges have a
    $\lpopt$-variable either set to $1$ or a value lower than $\mcfalpha$.
    In particular, by \Cref{th:lp_relaxation_bound}, there exist at least as
    many edges $e'$ with $\lpopt_{e'} = 1$ as there are edges~$e''$ with
    $\lpopt_{e''} \in (0,\mcfalpha)$.
    Thus, $|\{e \in E \ssep \lpopt_e = 1\}| \geq \frac{\algval-\Delta}{2}$.
    Hence, the optimal fractional solution $\lpopt$ has objective value
    \begin{align*}
        \lpval &\geq |\{e \in E  \ssep \lpopt_e = 1\}| + \mcfalpha \cdot |\{e \in E  \ssep \lpopt_e \in \mathopen[\mcfalpha, 1)\}|\\
               &\geq \frac{\algval-\Delta}{2} + \mcfalpha \cdot \Delta
               = \left(\frac{1}{2} + (\mcfalpha - \frac{1}{2}) \cdot
               \frac{\Delta}{\algval}\right) \cdot \algval.
    \end{align*}

    Using the minimum fractional objective value $\lpval$ as a lower bound for
    the minimum integral objective value~$\ilpval$, we can then give an upper
    bound for the approximation ratio:
    \begin{align*}
        r \be \frac{\algval}{\ilpval} \leq \frac{\algval}{\lpval}
        \leq \frac{1}{\frac{1}{2} + (\mcfalpha - \frac{1}{2}) \cdot \frac{\Delta}{\algval}}
    \end{align*}

    To obtain an upper bound on this ratio, we examine when its denominator
    is at its minimum.
    For $\mcfalpha \geq \frac{1}{2}$, the term $(\mcfalpha -
    \frac{1}{2})\cdot\frac{\Delta}{\algval}$ is always non-negative and reaches
    its lowest value~0 when $\Delta = 0$, giving us the upper bound~$2 \geq r$
    in this case.
    In contrast, for $\mcfalpha < \frac{1}{2}$, the term $(\mcfalpha -
    \frac{1}{2})\cdot\frac{\Delta}{\algval}$ is negative, and the minimum of the denominator is
    reached when~$\Delta = \algval$, leading to the upper
    bound~$\frac{1}{\mcfalpha} \geq r$ in this second case.
\end{proof}

It is surprisingly hard to find instances (in particular, ones without parallel
edges) where the worst-case approximation ratio given by~\Cref{th:mcfs_approx}
is actually met.
However, we show that such instances exist for the most relevant
$\mcfalpha \in (0,1)$, and hence our analysis is \emph{tight}:

\newcommand{\EXZ}{E_{XZ}}%
\newcommand{\EB}{E_{B}}%
\newcommand{\Ev}{E_{v}}%
\newcommand{\eXZ}{e_{XZ}}%
\newcommand{\ev}{e_{v}}%
\newcommand{\Eone}{E_1}%

\begin{lemma}\label{th:mcfs_approx_tightness_low}
    The ratio~$\nicefrac{1}{\mcfalpha}$ for the algorithm
    of \Cref{th:mcfs_approx} is tight for all~$\mcfalpha = \frac{1}{q}$,
    $q \in \N_{>1}$.
\end{lemma}
\begin{proof}
    Consider a complete bidirected graph~$G=(V,E)$ on $q+1$ vertices,
    and let~$H\subset G$ denote an arbitrarily chosen directed
    Hamiltonian cycle in~$G$ (an example for $q=4$ is
    visualized in \Cref{fig:mcps_lp_rounding_tight_low_alpha}).
    For all edges $st \in E$, let~$\eta(st)$ denote the number of edges of
    the unique $s$-$t$-path in $H$ and set the capacity~$\ecap(st) \be
    \frac{1}{\eta(st)}$.
    $E(H)$ is a feasible integral solution for the \prob{MMCFS}
    instance~$(G,\ecap,\mcfalpha)$ since the demands
    $\mcfalpha\cdot\alldemand$ are routable in it:
    A single edge~$e \in H$ is able to satisfy its own demand
    $\mcfalpha\cdot\ecap(e) = \frac{1}{q}\cdot 1$.
    Adding to this, for every $i \in \{2,\dots,q\}$, there are $i$~edges
    $uv\in E\setminus E(H)$ with $\eta(uv) = i$ whose unique
    $u$-$v$-path in $H$ contains~$e$.
    The edge~$e$ can accommodate all of these commodities~$uv$
    because each of them has a demand
    of~$\mcfalpha\cdot\ecap(uv)=\frac{1}{q}\cdot\frac{1}{i}$, summing up
    to a total flow of~$\sum_{i=1}^q i \cdot \frac{1}{i} \cdot
    \frac{1}{q} = 1 \leq \frac{1}{\ecap(e)}$.
    Lastly, $E(H)$ is not only feasible but optimal since we need at least
    $|E(H)|$ many edges to preserve the reachability relation of~$G$ in~$H$.

    The algorithm from \Cref{th:mcfs_approx} would potentially
    choose all edges of $G$, i.e., $|E| = (|V| - 1)\cdot|V| = q\cdot |E(H)| =
    \frac{1}{\mcfalpha}\cdot |E(H)|$ many:
    As the cost $\frac{1}{\ecap(uv)} = \eta(uv)$ of an edge~$uv\in
    E\setminus E(H)$ is equal to the total cost of its unique
    $u$-$v$-path in $H$, all $u$-$v$-paths are equal in cost.
    Thus, one optimal solution for the relaxation of ILP~\eqref{eq:ilp_mcfs}
    simply routes all commodities~$\comm\in E$ completely over their own edges
    $\comm$ and chooses~$\lpopt_{\comm} = \alpha$.
    This solution is also basic:
    Let~$\flow$ denote the flow variables of this solution.
    If it were not basic,
    there would exist two other optimal solutions with flow variables $\flow',
    \flow''$ respectively, such that for some~$\comm, e\in E$, $\flow'_{\comm}(e) <\flow_{\comm}(e) < \flow''_{\comm}(e)$.
    However, if $e = \comm$, $\flow$ already routes the full demand of
    commodity~$\comm$ over~$e$ and we cannot increase~$\flow_{\comm}(e)$
    without introducing a flow cycle, which would raise the objective value and
    thus be suboptimal.
    In contrast, if $e \neq \comm$, $\flow_{\comm}(e)$ is already~$0$ and
    cannot be decreased.
\end{proof}

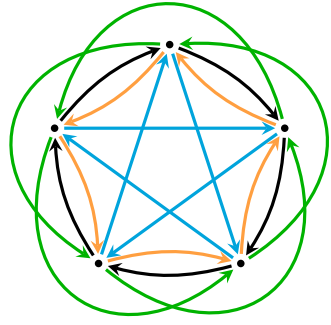
\begin{figure}[p]
    \centering
    \begin{tikzpicture}[scale=0.8]
        \begin{scope}[every node/.style={dot}]
            \node (a) at (234:2) {};
            \node (b) at (162:2) {};
            \node (c) at  (90:2) {};
            \node (d) at  (18:2) {};
            \node (e) at (306:2) {};
        \end{scope}

        \begin{scope}[every edge/.style={edge,bend left=15}]
            \path (a) edge (b);
            \path (b) edge (c);
            \path (c) edge (d);
            \path (d) edge (e);
            \path (e) edge (a);
        \end{scope}
        \begin{scope}[every edge/.style={edge,draw=myblue}]
            \path (a) edge (c);
            \path (c) edge (e);
            \path (e) edge (b);
            \path (b) edge (d);
            \path (d) edge (a);
        \end{scope}
        \begin{scope}[every path/.style={edge,draw=mygreen,bend right=60}]
            \path (a) .. controls (286:3.5) and (326:3.5) .. (d);
            \path (d) .. controls  (70:3.5) and (110:3.5) .. (b);
            \path (b) .. controls (214:3.5) and (254:3.5) .. (e);
            \path (e) .. controls  (-2:3.5) and  (38:3.5) .. (c);
            \path (c) .. controls (142:3.5) and (182:3.5) .. (a);
        \end{scope}
        \begin{scope}[every edge/.style={edge,bend left=15,draw=myorange}]
            \path (b) edge (a);
            \path (c) edge (b);
            \path (d) edge (c);
            \path (e) edge (d);
            \path (a) edge (e);
        \end{scope}
    \end{tikzpicture}
    \caption{\prob{MMCFS} instance for~$\mcfalpha = \frac{1}{4}$, constructed as a
        member of the family of instances in the proof
        of~\Cref{th:mcfs_approx_tightness_low}.
        The Hamiltonian cycle~$H$ is drawn as black.
        The value~$\eta(st)$, denoting the distance from~$s$ to~$t$ in~$H$ for
        an edge~$st$, is 2, 3, or 4 for blue, green, and orange edges,
        respectively.
    }
    \label{fig:mcps_lp_rounding_tight_low_alpha}
\end{figure}

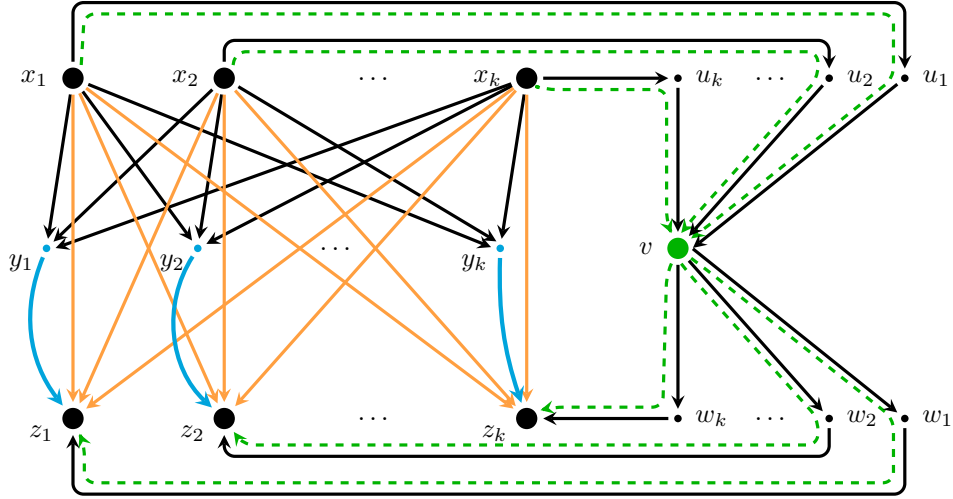
\begin{figure}[p]
    \centering
    \begin{tikzpicture}[scale=0.5]
        \begin{scope}[every node/.style={dot, minimum size=8pt}]
            \node[label={[label distance=1mm]180:$x_1$}] (x1) at (0,4.5) {};
            \node[label={[label distance=1mm]180:$x_2$}] (x2) at (4,4.5) {};
            \node[label={[label distance=1mm]180:$x_k$}] (xn) at (12,4.5) {};
            \node[label={[label distance=1mm]185:$z_1$}] (z1) at (0,-4.5) {};
            \node[label={[label distance=1mm]185:$z_2$}] (z2) at (4,-4.5) {};
            \node[label={[label distance=1mm]185:$z_k$}] (zn) at (12,-4.5) {};
            \node[fill=mygreen, label={[label distance=1mm]180:$v$}] (v) at (16,0) {};
        \end{scope}
        \begin{scope}[every node/.style={dot}]
            \node[fill=myblue,label={[label distance=1mm]185:$y_1$}] (y1) at (-0.7,0) {};
            \node[fill=myblue,label={[label distance=1mm]185:$y_2$}] (y2) at (3.3,0) {};
            \node[fill=myblue,label={[label distance=1mm]185:$y_k$}] (yn) at (11.3,0) {};
            \node[label={[label distance=1mm]0:$u_1$}] (u1) at (22,4.5) {};
            \node[label={[label distance=1mm]0:$u_2$}] (u2) at (20,4.5) {};
            \node[label={[label distance=1mm]0:$u_k$}] (un) at (16,4.5) {};
            \node[label={[label distance=1mm]0:$w_1$}] (w1) at (22,-4.5) {};
            \node[label={[label distance=1mm]0:$w_2$}] (w2) at (20,-4.5) {};
            \node[label={[label distance=1mm]0:$w_k$}] (wn) at (16,-4.5) {};
        \end{scope}

        \begin{scope}[every node/.style={circle,minimum size=1pt,inner sep=1pt}]
            \node (xdots) at (8,4.5) {$\dots$};
            \node (zdots) at (8,-4.5) {$\dots$};
            \node (ydots) at (7,0) {$\dots$};
            \node (udots) at (18.5,4.5) {$\dots$};
            \node (wdots) at (18.5,-4.5) {$\dots$};
        \end{scope}

        \begin{scope}[every path/.style={edge,draw=mygreen, very thick, dashed}]
            \draw (x1.45) -- ($(x1)+(0.3,1.7)$) -- ($(u1)+(-0.3,1.7)$) -- ($(u1)+(-0.3,0)$) -- ($(v)+(0.8,0.6)$) -- (v);
            \draw (x2.45) -- ($(x2)+(0.3,0.7)$) -- ($(u2)+(-0.3,0.7)$) -- ($(u2)+(-0.3,0)$) -- ($(v)+(0.4,0.8)$) -- (v.90);
            \draw (xn.-45) -- ($(xn)+(0.5,-0.3)$) -- ($(un)+(-0.3,-0.3)$) -- ($(v)+(-0.3,0.7)$) -- (v);
            \path (v) -- ($(v)+(0.8,-0.6)$) -- ($(w1)+(-0.3,0)$) -- ($(w1)+(-0.3,-1.7)$) -- ($(z1)+(0.3,-1.7)$) -- ($(z1)+(0.3,-0.7)$) -- (z1);
            \path (v.-90) -- ($(v)+(0.4,-0.8)$) -- ($(w2)+(-0.3,0)$) -- ($(w2)+(-0.3,-0.7)$) -- ($(z2)+(0.5,-0.7)$) -- (z2);
            \path (v) -- ($(v)+(-0.3,-0.7)$) -- ($(wn)+(-0.4,0.3)$) -- ($(zn)+(0.7,0.3)$) -- (zn.40);
        \end{scope}

        \begin{scope}[every edge/.style={edge}]
            \path (x1) edge (y1); \path (x1) edge (y2); \path (x1) edge (yn.190);
            \path (x2) edge (y1); \path (x2) edge (y2); \path (x2) edge (yn.80);
            \path (xn) edge (y1.-20); \path (xn) edge (y2); \path (xn) edge (yn);
        \end{scope}
        \begin{scope}[every edge/.style={edge, draw=myorange, very thick}]
            \path (x1) edge (z1); \path (x1) edge (z2); \path (x1) edge (zn);
            \path (x2) edge (z1); \path (x2) edge (z2); \path (x2) edge (zn);
            \path (xn) edge (z1); \path (xn) edge (z2); \path (xn) edge (zn);
        \end{scope}
        \begin{scope}[every edge/.style={edge, draw=myblue, ultra thick}]
            \path (y1) edge[bend right=30] (z1);
            \path (y2) edge[bend right=40] (z2);
            \path (yn) edge[bend right=10] (zn);
        \end{scope}
        \begin{scope}[every path/.style={edge}]
            \draw (x1) -- ($(x1)+(0,2)$) -- ($(u1)+(0,2)$) -- (u1);
            \draw (x2) -- ($(x2)+(0,1)$) -- ($(u2)+(0,1)$) -- (u2);
            \draw (xn) -- (un);
            \draw (u1) -- (v.-15);
            \draw (u2) -- (v);
            \draw (un) -- (v);
            \draw (w1) -- ($(w1)-(0,2)$) -- ($(z1)-(0,2)$) -- (z1);
            \draw (w2) -- ($(w2)-(0,1)$) -- ($(z2)-(0,1)$) -- (z2);
            \draw (wn) -- (zn);
            \draw (v.15) -- (w1);
            \draw (v) -- (w2);
            \draw (v) -- (wn);
        \end{scope}
    \end{tikzpicture}
    \caption{Family of \prob{MMCFS} instances constructed in the proof
    of~\Cref{th:mcfs_approx_tightness_high}.
    Edges in $\Ev$ are drawn in green (and dashed), edges in $\EXZ$ are orange,
    edges in $\Eone$ are black, and edges in $\EB$ are blue.
    Recall that edges $\Eone$ and $\EB$ (black and blue) form the unique
    \ac{MED}, which is contained in every feasible solution.
    There exists a feasible \prob{MMCFS} solution~$\Eone \union \EB \union \Ev$
    without any (orange) edges from~$\EXZ$, but the algorithm from
    \Cref{th:mcfs_approx} will include all edges of~$\EXZ$ in its solution.}
    \label{fig:mcps_lp_rounding_tight_high_alpha}
\end{figure}

\begin{lemma}\label{th:mcfs_approx_tightness_high}
    The ratio~$2$ for the algorithm
    of \Cref{th:mcfs_approx} is tight for all~$\mcfalpha > \frac{1}{2}$.
\end{lemma}
\begin{proof}
    We construct a family of \prob{MMCFS} instances (visualized in
    \Cref{fig:mcps_lp_rounding_tight_high_alpha}), each one consisting
    of a simple graph~$G=(V,E)$ with edge capacities~$\ecap$ and the
    given retention ratio~$\mcfalpha$, that meets the approximation
    ratio of 2 asymptotically with increasing size.
    Intuitively, the edge set of each instance contains two subsets that have a
    size quadratic in the number of vertices and two subsets of linear size; our
    algorithm chooses both quadratic-size subsets even though one could be
    replaced by a linear-size subset.

    Let $V$ be comprised of five disjoint vertex sets $X$, $Y$, $Z$, $U$, $W$
    with $k \in \N$ vertices each (respectively named with the corresponding
    lowercase letter and indexed by $i \in \{1,\dots,k\}$) and a distinct
    vertex~$v$.
    Further, let $\varepsilon \leq \frac{1-\mcfalpha}{k\cdot\mcfalpha}$ and $B
    \geq \frac{k\cdot\varepsilon}{1-\mcfalpha}$ be sufficiently low and high
    positive values, respectively.
    Then, $E \be \Ev \union \EXZ \union \Eone \union \EB$ with
    \begin{align*}
        \Ev &\be (X \times \{v\}) \union (\{v\} \times Z) &
        \EXZ &\be X \times Z\\
        \Eone &\be X \times Y \union \{x_iu_i,u_iv,vw_i,w_iz_i \ssep i \in
        \{1,\dots,k\}\} & \EB &\be \{y_iz_i \ssep i \in \{1,\dots,k\}\}\\
        \ecap(e) &\be
        \mathrlap{\Big\{{\textstyle \frac{1-\mcfalpha}{\mcfalpha}}~\text{if}~e \in \Ev,\qquad
                {\textstyle \frac{1-\mcfalpha+\varepsilon}{\mcfalpha}}~\text{if}~e \in \EXZ,\qquad
        1~\text{if}~e \in \Eone,\qquad
        B~\text{if}~e \in \EB.}
    \end{align*}

    $G$ is a \ac{DAG}, thus its \ac{MED} is
    unique~\cite{DBLP:journals/siamcomp/AhoGU72} and must be part of any
    feasible solution for the \prob{MMCFS}
    instance~$(G,\ecap,\mcfalpha)$, see \Cref{th:med_in_mcfs}.
    The \ac{MED} of~$G$ consists of the edges~$\Eone \union \EB$ and in
    particular can satisfy all demands~$\mcfalpha\cdot\alldemand(e) =
    \mcfalpha\cdot\ecap(e)$ with $e \in \Eone \union \EB$.
    The \enquote{remaining} capacity of~$1-\mcfalpha$ for each of the edges
    in~$\Eone$ (and $(1-\mcfalpha) \cdot B$ for edges in $\EB$) is
    also sufficient to satisfy the demands $\mcfalpha\cdot\alldemand(\ev)$
    of the commodities $\ev \in \Ev$, and to \emph{almost} satisfy the
    demands $\mcfalpha\cdot\alldemand(\eXZ) = 1 - \mcfalpha +
    \varepsilon$ for the commodities~$\eXZ \in \EXZ$:
    it only leaves a demand of~$\varepsilon$ for every such~$\eXZ$.
    Thus, there exists a feasible solution~$\Eone \union \EB \union \Ev$ for
    $(G,\ecap,\mcfalpha)$ that contains all of the $k^2+5k$~edges from~$\Eone
    \union \EB$ and additionally, to satisfy the remaining demands (of
    size~$\varepsilon$ each), the $2k$~edges from~$\Ev$.
    This feasible solution gives an upper bound on the objective
    value~$\ilpval$ of the optimal integral solution: $\ilpval \leq k^2
    +7k$.

    In contrast, the algorithm from \Cref{th:mcfs_approx} also includes
    all $k^2+5k$~edges of the \ac{MED}~$\Eone \union \EB$ in its
    solution, but must additionally choose (at least) the $k^2$~edges
    from~$\EXZ$, resulting in an objective value~$\algval \geq 2k^2 +5k$. This
    is because every optimal solution for the relaxation of
    ILP~\eqref{eq:ilp_mcfs} routes the commodities $st=\eXZ \in \EXZ$ over the
    edges~$\eXZ$ themselves: they have a lower cost than the alternative
    $s$-$t$-path of length 2 over the edges in~$\Ev$.

    For increasing~$k$, the ratio of the algorithmic solution's objective value
    to the optimum is
    \begin{align*}
        \lim_{k \to \infty} \frac{\algval}{\ilpval} \geq \lim_{k \to
        \infty}\frac{2k^2+5k}{k^2+7k} = 2.&\qedhere
    \end{align*}
\end{proof}

Combining \Cref{th:mcfs_approx} with \Cref{th:mcfs_approx_tightness_low,th:mcfs_approx_tightness_high}, we obtain:

\begin{corollary}\label{th:mcfs_approx_tightness}
    The approximation ratio~$\max(\nicefrac{1}{\mcfalpha}, 2)$ for the algorithm
    given by \Cref{th:mcfs_approx} is tight for all~$\mcfalpha > \frac{1}{2}$
    and all $\mcfalpha = \frac{1}{q}$ with~$q \in \N_{>1}$.
\end{corollary}

\begin{remark}
    There are instances where the integrality gap of
    ILP~\eqref{eq:ilp_mcfs} is $\frac{1}{\mcfalpha}$---e.g.\ trees, where the
    unique optimal integral solution contains every edge while the unique optimal
    fractional solution chooses $\lpopt_e = \mcfalpha$ for all edges~$e$ of the input
    graph.
    Interestingly, the approximation ratio of our algorithm (for $\mcfalpha \leq
    \frac{1}{2}$) equals the integrality gap exactly, but on completely different
    instances and not on trees (where the algorithm always finds the optimal
    solution).
\end{remark}

\section{Conclusion and Open Questions}\label{sec:conclusion}

We introduced the practically motivated \probf{MMCFS} problem and paved the way
for further research by giving
several structural results, most importantly a reformulation of this traffic-oblivious
problem that only needs to consider a single specific traffic matrix.
Further, we showed that \prob{MMCFS} is NP-hard (and a closely related problem
cannot be approximated within a sublogarithmic factor) already on \acp{DAG}.
Lastly, we gave an extremely simple LP-rounding scheme for
\prob{MMCFS} with a tight approximation guarantee of~$\max(\nicefrac{1}{\mcfalpha}, 2)$.

Considering seemingly related problems (see \Cref{sec:relwork}), one observes
that an approximation ratio of 2 (which we attain for the practically most
relevant cases of $\mcfalpha\geq\frac12$~\cite{DBLP:conf/lcn/OttenICA23}) often
arises as a seemingly ``natural'' limit for such ratios.
Yet, it remains an open question whether there exists an approximation algorithm
for \prob{MMCFS} with a better quality guarantee, and whether there is a
non-trivial lower bound on the approximation guarantee for any such algorithm
(assuming $\text{P} \neq \text{NP}$).
Further, it might be of interest to explore several generalizations of
\prob{MMCFS}: this includes the non-traffic-oblivious variant where a specific
traffic matrix is part of the input (which is also NP-hard via
\Cref{th:super_t}), and the variant where, given an additional cost function on
the edges, one asks for a subgraph of minimum~cost.

\newpage

\bibliography{main_reduced}

\begin{thebibliography}{10}

\bibitem{DBLP:journals/csr/AhmedBSHJKS20}
Abu~Reyan Ahmed, Greg Bodwin, Faryad~Darabi Sahneh, Keaton Hamm, Mohammad Javad~Latifi Jebelli, Stephen~G. Kobourov, and Richard Spence.
\newblock Graph spanners: {A} tutorial review.
\newblock {\em Comput. Sci. Rev.}, 37:100253, 2020.
\newblock \href {https://doi.org/10.1016/J.COSREV.2020.100253} {\path{doi:10.1016/J.COSREV.2020.100253}}.

\bibitem{DBLP:journals/siamcomp/AhoGU72}
Alfred~V. Aho, Michael~R. Garey, and Jeffrey~D. Ullman.
\newblock The transitive reduction of a directed graph.
\newblock {\em {SIAM} J. Comput.}, 1(2):131--137, 1972.
\newblock \href {https://doi.org/10.1137/0201008} {\path{doi:10.1137/0201008}}.

\bibitem{DBLP:books/daglib/0069809}
Ravindra~K. Ahuja, Thomas~L. Magnanti, and James~B. Orlin.
\newblock {\em Network flows - theory, algorithms and applications}.
\newblock Prentice Hall, 1993.

\bibitem{DBLP:journals/tcs/Al-NajjarBL21}
Yacine Al{-}Najjar, Walid Ben{-}Ameur, and J{\'{e}}r{\'{e}}mie Leguay.
\newblock On the approximability of robust network design.
\newblock {\em Theor. Comput. Sci.}, 860:41--50, 2021.
\newblock \href {https://doi.org/10.1016/J.TCS.2021.01.026} {\path{doi:10.1016/J.TCS.2021.01.026}}.

\bibitem{DBLP:conf/stacs/Al-NajjarBL22}
Yacine Al{-}Najjar, Walid Ben{-}Ameur, and J{\'{e}}r{\'{e}}mie Leguay.
\newblock Approximability of robust network design: The directed case.
\newblock In {\em Proc. {STACS} 2022}, volume 219 of {\em LIPIcs}, pages 6:1--6:16. Schloss Dagstuhl - Leibniz-Zentrum f{\"{u}}r Informatik, 2022.
\newblock \href {https://doi.org/10.4230/LIPICS.STACS.2022.6} {\path{doi:10.4230/LIPICS.STACS.2022.6}}.

\bibitem{DBLP:journals/corr/abs-0907-3631}
Reid Andersen and Uriel Feige.
\newblock Interchanging distance and capacity in probabilistic mappings.
\newblock {\em CoRR}, abs/0907.3631, 2009.
\newblock \href {https://arxiv.org/abs/0907.3631} {\path{arXiv:0907.3631}}.

\bibitem{DBLP:conf/soda/AndoniGK14}
Alexandr Andoni, Anupam Gupta, and Robert Krauthgamer.
\newblock Towards ({1} + $\epsilon$)-approximate flow sparsifiers.
\newblock In {\em Proc. {SODA} 2014}, pages 279--293. {SIAM}, 2014.
\newblock \href {https://doi.org/10.1137/1.9781611973402.20} {\path{doi:10.1137/1.9781611973402.20}}.

\bibitem{DBLP:conf/waoa/Antonakopoulos10}
Spyridon Antonakopoulos.
\newblock Approximating directed buy-at-bulk network design.
\newblock In {\em Proc. {WAOA} 2010}, volume 6534 of {\em LNCS}, pages 13--24. Springer, 2010.
\newblock \href {https://doi.org/10.1007/978-3-642-18318-8\_2} {\path{doi:10.1007/978-3-642-18318-8\_2}}.

\bibitem{DBLP:journals/jcss/AzarCFKR04}
Yossi Azar, Edith Cohen, Amos Fiat, Haim Kaplan, and Harald R{\"{a}}cke.
\newblock Optimal oblivious routing in polynomial time.
\newblock {\em J. Comput. Syst. Sci.}, 69(3):383--394, 2004.
\newblock \href {https://doi.org/10.1016/J.JCSS.2004.04.010} {\path{doi:10.1016/J.JCSS.2004.04.010}}.

\bibitem{ben2005routing}
Walid Ben{-}Ameur and Herv{\'e} Kerivin.
\newblock Routing of uncertain traffic demands.
\newblock {\em Optimization and Engineering}, 6:283--313, 2005.
\newblock \href {https://doi.org/10.1145/777313.777314} {\path{doi:10.1145/777313.777314}}.

\bibitem{DBLP:conf/wads/BermanDK09}
Piotr Berman, Bhaskar DasGupta, and Marek Karpinski.
\newblock Approximating transitive reductions for directed networks.
\newblock In {\em Proc. {WADS} 2009}, volume 5664 of {\em LNCS}, pages 74--85. Springer, 2009.
\newblock \href {https://doi.org/10.1007/978-3-642-03367-4\_7} {\path{doi:10.1007/978-3-642-03367-4\_7}}.

\bibitem{DBLP:conf/infocom/BhatiaHKL15}
Randeep Bhatia, Fang Hao, Murali~S. Kodialam, and T.~V. Lakshman.
\newblock Optimized network traffic engineering using segment routing.
\newblock In {\em Proc. {INFOCOM} 2015}, pages 657--665. {IEEE}, 2015.
\newblock \href {https://doi.org/10.1109/INFOCOM.2015.7218434} {\path{doi:10.1109/INFOCOM.2015.7218434}}.

\bibitem{DBLP:journals/siamcomp/ChekuriHKS10}
Chandra Chekuri, Mohammad~Taghi Hajiaghayi, Guy Kortsarz, and Mohammad~R. Salavatipour.
\newblock Approximation algorithms for nonuniform buy-at-bulk network design.
\newblock {\em {SIAM} J. Comput.}, 39(5):1772--1798, 2010.
\newblock \href {https://doi.org/10.1137/090750317} {\path{doi:10.1137/090750317}}.

\bibitem{DBLP:conf/icc/ChiaraviglioMN09}
Luca Chiaraviglio, Marco Mellia, and Fabio Neri.
\newblock Reducing power consumption in backbone networks.
\newblock In {\em Proc. {ICC} 2009}, pages 1--6. {IEEE}, 2009.
\newblock \href {https://doi.org/10.1109/ICC.2009.5199404} {\path{doi:10.1109/ICC.2009.5199404}}.

\bibitem{DBLP:conf/iwoca/ChimaniI23}
Markus Chimani and Max Ilsen.
\newblock Capacity-preserving subgraphs of directed flow networks.
\newblock In {\em Proc. {IWOCA} 2023}, volume 13889 of {\em LNCS}, pages 160--172. Springer, 2023.
\newblock \href {https://doi.org/10.1007/978-3-031-34347-6\_14} {\path{doi:10.1007/978-3-031-34347-6\_14}}.

\bibitem{DBLP:conf/stoc/DinurS14}
Irit Dinur and David Steurer.
\newblock Analytical approach to parallel repetition.
\newblock In {\em Proc. {STOC} 2014}, pages 624--633. {ACM}, 2014.
\newblock \href {https://doi.org/10.1145/2591796.2591884} {\path{doi:10.1145/2591796.2591884}}.

\bibitem{DBLP:journals/siamcomp/EnglertGKRTT14}
Matthias Englert, Anupam Gupta, Robert Krauthgamer, Harald R{\"{a}}cke, Inbal Talgam{-}Cohen, and Kunal Talwar.
\newblock Vertex sparsifiers: New results from old techniques.
\newblock {\em {SIAM} J. Comput.}, 43(4):1239--1262, 2014.
\newblock \href {https://doi.org/10.1137/130908440} {\path{doi:10.1137/130908440}}.

\bibitem{DBLP:journals/rfc/rfc8402}
Clarence Filsfils, Stefano Previdi, Les Ginsberg, Bruno Decraene, Stephane Litkowski, and Rob Shakir.
\newblock Segment routing architecture.
\newblock {\em {RFC}}, 8402:1--32, 2018.
\newblock \href {https://doi.org/10.17487/RFC8402} {\path{doi:10.17487/RFC8402}}.

\bibitem{Foulds1981273}
Les~R. Foulds.
\newblock A multi-commodity flow network design problem.
\newblock {\em Transportation Research Part B: Methodological}, 15(4):273--283, 1981.
\newblock \href {https://doi.org/10.1016/0191-2615(81)90013-8} {\path{doi:10.1016/0191-2615(81)90013-8}}.

\bibitem{DBLP:books/fm/GareyJ79}
Michael~R. Garey and David~S. Johnson.
\newblock {\em Computers and Intractability: {A} Guide to the Theory of NP-Completeness}.
\newblock W. H. Freeman, 1979.

\bibitem{Gendron1999}
Bernard Gendron, Teodor~Gabriel Crainic, and Antonio Frangioni.
\newblock {\em Multicommodity Capacitated Network Design}, pages 1--19.
\newblock Springer US, Boston, MA, 1999.
\newblock \href {https://doi.org/10.1007/978-1-4615-5087-7_1} {\path{doi:10.1007/978-1-4615-5087-7_1}}.

\bibitem{DBLP:journals/ejco/GendronL14}
Bernard Gendron and Mathieu Larose.
\newblock Branch-and-price-and-cut for large-scale multicommodity capacitated fixed-charge network design.
\newblock {\em {EURO} J. Comput. Optim.}, 2(1-2):55--75, 2014.
\newblock \href {https://doi.org/10.1007/S13675-014-0020-9} {\path{doi:10.1007/S13675-014-0020-9}}.

\bibitem{DBLP:journals/talg/HajiaghayiKRL07}
Mohammad~Taghi Hajiaghayi, Robert~D. Kleinberg, Harald R{\"{a}}cke, and Tom Leighton.
\newblock Oblivious routing on node-capacitated and directed graphs.
\newblock {\em {ACM} Trans. Algorithms}, 3(4):51, 2007.
\newblock \href {https://doi.org/10.1145/1290672.1290688} {\path{doi:10.1145/1290672.1290688}}.

\bibitem{DBLP:conf/cp/HartertSVB15}
Renaud Hartert, Pierre Schaus, Stefano Vissicchio, and Olivier Bonaventure.
\newblock Solving segment routing problems with hybrid constraint programming techniques.
\newblock In {\em Proc. {CP} 2015}, volume 9255 of {\em LNCS}, pages 592--608. Springer, 2015.
\newblock \href {https://doi.org/10.1007/978-3-319-23219-5\_41} {\path{doi:10.1007/978-3-319-23219-5\_41}}.

\bibitem{DBLP:journals/combinatorica/Jain01}
Kamal Jain.
\newblock A factor 2 approximation algorithm for the generalized steiner network problem.
\newblock {\em Comb.}, 21(1):39--60, 2001.
\newblock \href {https://doi.org/10.1007/s004930170004} {\path{doi:10.1007/s004930170004}}.

\bibitem{DBLP:conf/coco/Karp72}
Richard~M. Karp.
\newblock Reducibility among combinatorial problems.
\newblock In {\em Proc. {COCO} 1972}, The {IBM} Research Symposia Series, pages 85--103. Plenum Press, New York, 1972.
\newblock \href {https://doi.org/10.1007/978-1-4684-2001-2\_9} {\path{doi:10.1007/978-1-4684-2001-2\_9}}.

\bibitem{DBLP:journals/ipl/KhullerRY94}
Samir Khuller, Balaji Raghavachari, and Neal~E. Young.
\newblock Designing multi-commodity flow trees.
\newblock {\em Inf. Process. Lett.}, 50(1):49--55, 1994.
\newblock \href {https://doi.org/10.1016/0020-0190(94)90044-2} {\path{doi:10.1016/0020-0190(94)90044-2}}.

\bibitem{DBLP:journals/siamcomp/KhullerRF95}
Samir Khuller, Balaji Raghavachari, and Neal~E. Young.
\newblock Approximating the minimum equivalent digraph.
\newblock {\em {SIAM} J. Comput.}, 24(4):859--872, 1995.
\newblock \href {https://doi.org/10.1137/S0097539793256685} {\path{doi:10.1137/S0097539793256685}}.

\bibitem{DBLP:conf/stoc/LeightonM10}
Frank~Thomson Leighton and Ankur Moitra.
\newblock Extensions and limits to vertex sparsification.
\newblock In {\em Proc. {STOC} 2010}, pages 47--56. {ACM}, 2010.
\newblock \href {https://doi.org/10.1145/1806689.1806698} {\path{doi:10.1145/1806689.1806698}}.

\bibitem{DBLP:journals/networks/MelkonianT04}
Vardges Melkonian and {\'{E}}va Tardos.
\newblock Algorithms for a network design problem with crossing supermodular demands.
\newblock {\em Networks}, 43(4):256--265, 2004.
\newblock \href {https://doi.org/10.1002/NET.20005} {\path{doi:10.1002/NET.20005}}.

\bibitem{DBLP:conf/focs/Moitra09}
Ankur Moitra.
\newblock Approximation algorithms for multicommodity-type problems with guarantees independent of the graph size.
\newblock In {\em Proc. {FOCS} 2009}, pages 3--12. {IEEE} Computer Society, 2009.
\newblock \href {https://doi.org/10.1109/FOCS.2009.28} {\path{doi:10.1109/FOCS.2009.28}}.

\bibitem{DBLP:journals/toc/Moshkovitz15}
Dana Moshkovitz.
\newblock The projection games conjecture and the np-hardness of ln n-approxi\-mating set-cover.
\newblock {\em Theory Comput.}, 11:221--235, 2015.
\newblock \href {https://doi.org/10.4086/toc.2015.v011a007} {\path{doi:10.4086/toc.2015.v011a007}}.

\bibitem{DBLP:conf/lcn/OttenICA23}
Daniel Otten, Max Ilsen, Markus Chimani, and Nils Aschenbruck.
\newblock Green traffic engineering by line card minimization.
\newblock In {\em Proc. {LCN} 2023}, pages 1--8. {IEEE}, 2023.
\newblock \href {https://doi.org/10.1109/LCN58197.2023.10223344} {\path{doi:10.1109/LCN58197.2023.10223344}}.

\bibitem{DBLP:conf/focs/Racke02}
Harald R{\"{a}}cke.
\newblock Minimizing congestion in general networks.
\newblock In {\em Proc. {FOCS} 2002}, pages 43--52. {IEEE} Computer Society, 2002.
\newblock \href {https://doi.org/10.1109/SFCS.2002.1181881} {\path{doi:10.1109/SFCS.2002.1181881}}.

\bibitem{DBLP:conf/stoc/Racke08}
Harald R{\"{a}}cke.
\newblock Optimal hierarchical decompositions for congestion minimization in networks.
\newblock In {\em Proc. {STOC} 2008}, pages 255--264. {ACM}, 2008.
\newblock \href {https://doi.org/10.1145/1374376.1374415} {\path{doi:10.1145/1374376.1374415}}.

\bibitem{DBLP:journals/siamcomp/Sahni74}
Sartaj Sahni.
\newblock Computationally\,related\,problems.
\newblock {\em {SIAM} J. Comput.}, 3(4):262--279, 1974.
\newblock \href {https://doi.org/10.1137/0203021} {\path{doi:10.1137/0203021}}.

\bibitem{DBLP:journals/or/SalimifardB22}
Khodakaram Salimifard and Sara Bigharaz.
\newblock The multicommodity network flow problem: state of the art classification, applications, and solution methods.
\newblock {\em Oper. Res.}, 22(1):1--47, 2022.
\newblock \href {https://doi.org/10.1007/S12351-020-00564-8} {\path{doi:10.1007/S12351-020-00564-8}}.

\bibitem{DBLP:conf/soda/SalmanCRS97}
F.~Sibel Salman, Joseph Cheriyan, R.~Ravi, and S.~Subramanian.
\newblock Buy-at-bulk network design: Approximating the single-sink edge installation problem.
\newblock In {\em Proc. {SODA} 1997}, pages 619--628. {ACM/SIAM}, 1997.
\newblock URL: \url{http://dl.acm.org/citation.cfm?id=314161.314397}.

\bibitem{DBLP:journals/ton/SchullerACHS18}
Timmy Sch{\"{u}}ller, Nils Aschenbruck, Markus Chimani, Martin Horneffer, and Stefan Schnitter.
\newblock Traffic engineering using segment routing and considering requirements of a carrier {IP} network.
\newblock {\em {IEEE/ACM} Trans. Netw.}, 26(4):1851--1864, 2018.
\newblock \href {https://doi.org/10.1109/TNET.2018.2854610} {\path{doi:10.1109/TNET.2018.2854610}}.

\bibitem{DBLP:books/daglib/0004338}
Vijay~V. Vazirani.
\newblock {\em Approximation algorithms}.
\newblock Springer, 2001.
\newblock URL: \url{http://www.springer.com/computer/theoretical+computer+science/book/978-3-540-65367-7}.

\bibitem{DBLP:conf/soda/Vetta01}
Adrian Vetta.
\newblock Approximating the minimum strongly connected subgraph via a matching lower bound.
\newblock In {\em Proc. {SODA} 2001}, pages 417--426. {ACM/SIAM}, 2001.
\newblock URL: \url{http://dl.acm.org/citation.cfm?id=365411.365493}.

\bibitem{DBLP:books/daglib/0030297}
David~P. Williamson and David~B. Shmoys.
\newblock {\em The Design of Approximation Algorithms}.
\newblock Cambridge University Press, 2011.
\newblock URL: \url{http://www.cambridge.org/de/knowledge/isbn/item5759340/}.

\bibitem{DBLP:conf/icnp/ZhangYLZ10}
Mingui Zhang, Cheng Yi, Bin Liu, and Beichuan Zhang.
\newblock {GreenTE}: Power-aware traffic engineering.
\newblock In {\em Proc. {ICNP} 2010}, pages 21--30. {IEEE} Computer Society, 2010.
\newblock \href {https://doi.org/10.1109/ICNP.2010.5762751} {\path{doi:10.1109/ICNP.2010.5762751}}.

\bibitem{DBLP:journals/ipl/ZhaoNI03}
Liang Zhao, Hiroshi Nagamochi, and Toshihide Ibaraki.
\newblock A linear time 5/3-approximation for the minimum strongly-connected spanning subgraph problem.
\newblock {\em Inf. Process. Lett.}, 86(2):63--70, 2003.
\newblock \href {https://doi.org/10.1016/S0020-0190(02)00476-3} {\path{doi:10.1016/S0020-0190(02)00476-3}}.

\end{thebibliography}

\end{document}